\documentclass[oribibl]{llncs}

\usepackage{amssymb}
\usepackage{amsmath}
\usepackage[linesnumbered,ruled,vlined,english]{algorithm2e}
\usepackage{fullpage}
\usepackage{graphicx}
\usepackage{subfigure}
\usepackage{stmaryrd}
\usepackage{url}
\usepackage[linesnumbered,ruled,vlined,english]{algorithm2e}
\usepackage{amsfonts}
\usepackage{amsmath}
\usepackage{amssymb}
\usepackage{epstopdf}
\usepackage[normalem]{ulem}
\usepackage[numbers]{natbib}
\usepackage[usenames,dvipsnames]{color}
\usepackage{xcolor}
\usepackage[space]{grffile}
\usepackage{verbatim}
\usepackage{authblk}
\usepackage{natbib}
\usepackage{booktabs}
\usepackage{array}

\newcommand{\cT}{\mathcal T}
\newcommand{\cL}{\mathcal L}

\newcommand{\cF}{\mathcal F}

\newcommand{\cX}{\mathcal X}

\newcommand{\scAG}{\textsc{AG}}

\newcommand{\scLCA}{\textsc{LCA}}

\newcommand{\scAMa}{\textsc{allMAAFs}$^1$}
\newcommand{\scAMSa}{\textsc{allMAFs}$^1$}
\newcommand{\scAMc}{\textsc{allMAAFs}$^2$}

\newtheorem{observation}{Observation}

\def\fc{$\ensuremath{\wedge}$}
\def\fcL{\bigwedge}
\title{Computing a Relevant Set of Nonbinary Maximum Acyclic Agreement Forests} 

\author{Benjamin Albrecht
\thanks{
Benjamin Albrecht\\
Institut f\"ur Informatik, Ludwig-Maximilians-Universit\"at, Germany\\
Tel.: +49-89-2180-4069\\
E-mail: \email{albrecht@bio.ifi.lmu.de}}
}

\institute{Institut f\"ur Informatik, Ludwig-Maximilians-Universit\"at, Germany}
\date{\today}

\begin{document}
\maketitle
\vspace{3cm}

\pagestyle{plain}
\pagenumbering{arabic}

\begin{abstract}
There exist several methods dealing with the reconstruction of rooted phylogenetic networks explaining different evolutionary histories given by rooted binary phylogenetic trees. In practice, however, due to insufficient information of the underlying data, phylogenetic trees are in general not completely resolved and, thus, those methods can often not be applied to biological data. In this work, we make a first important step to approach this goal by presenting the first algorithm --- called \textsc{allMulMAAFs} --- that enables the computation of all relevant nonbinary maximum acyclic agreement forests for two rooted (nonbinary) phylogenetic trees on the same set of taxa. Notice that our algorithm is part of the freely available software Hybroscale computing minimum hybridization networks for a set of rooted (nonbinary) phylogenetic trees on an overlapping set of taxa.\\

\textbf{Keywords} Directed Acyclic Graphs $\cdot$ Hybridization $\cdot$ Agreement Forests $\cdot$ Bounded Search $\cdot$ Phylogenetics

\end{abstract}

\newpage
\section{Introduction}

Evolution is often modeled by a rooted phylogenetic tree, which is a rooted tree whose edges are directed from its root to its leaf. The leaf set of such a tree is typically labeled by a set of species or strains of species and is in general called taxa set. Moreover, each internal node has in-degree one and out-degree two representing a certain speciation event. In applied phylogenetics, however, trees can contain nodes of out-degree larger or equal to three because, often, in order to resolve a certain ordering of speciation events, there is only insufficient information available and the common way to model this uncertainty is to use \emph{nonbinary} nodes. This means, in particular, that \emph{nonbinary} rooted phylogenetic trees contain several binary trees each each representing a different order of speciation events. 

Note that there are two ways of interpreting a nonbinary node, which is also called \emph{multifurcation} or \emph{polytomy}. A multifurcation is \emph{soft}, if it represents a lack of resolution of the true relationship as already mentioned above. A \emph{hard} multifurcation, in contrast, represent a simultaneous speciation event of at least three species. However, as such events are assumed to be rare in nature, in this paper we consider all multifurcations of a pyhlogenetic tree as being soft.

A phylogenetic tree is usually built on homologous genes corresponding to a certain set of species. Thus, when constructing two phylogenetic trees based on two different homologous genes each belonging to the same set of species, one usually expects that both inferred trees are identical. Due to biological reasons, apart from other more technical reasons, however, these trees can be incongruent. One of those biological reasons are reticulation events that are processes including hybridization, horizontal gene transfer, or recombination. Now, based on such trees, one is interested to investigate the underlying reticulate evolution which can be done, for example, by computing rooted phylogenetic networks being structures similar to phylogenetic trees but, in contrast, can contain nodes of in-degree larger or equal to two. In a biological context, each of those nodes in such a network represents a certain reticulation event and, thus, a node with multiple incoming edges is called reticulation node or, in respect of hybridization, hybridization node. 

Given two rooted binary phylogenetic trees $T_1$ and $T_2$ on the same set of taxa $\cX$, a phylogenetic network, also called hybridization network, can be calculated by applying two major steps. In a first step, one can compute an agreement forest consisting of edge disjoint subtrees each being part of both input trees. Moreover, this set of subtrees $\cF$ has to be a partition on $\cX$ that is not allowed to have any conflicting ancestral relations in terms of $T_1$ and $T_2$. This means, in particular, for each pair of subtrees $F_1$ and $F_2$ in $\cF$, if regarding $T_1$ the subtree corresponding to $F_1$ lies above or below the subtree corresponding to $F_2$, the same scenario has to be displayed in $T_2$. Notice that, by saying a subtree $F_1$ lies above $F_2$ we mean that there is a directed path leading from the root node corresponding to $F_1$ to the root node corresponding to $F_2$. Moreover, if this property holds for each pair of subtrees, we say that $\cF$ is acyclic. Now, based on such an acyclic agreement forest, in a second step, one can apply a network construction algorithm, e.g., the algorithm $\textsc{HybridPhylogeny}$ \cite{Baroni2006}, which glues together each component by introducing reticulation nodes such that the resulting network displays both input trees. If such a network additionally provides a minimum number of reticulation events, it is called a minimum hybridization network. Notice that, as this is a highly combinatorial problem, for two rooted binary phylogenetic trees there typically exist multiple acyclic agreement forests and a hybridization network corresponding to one of those agreement forests is rarely unique.

Computing maximum acyclic agreement forests, which are acyclic agreement forests of minimum cardinality, is a NP-hard as well as APX-hard problem \cite{Bordewich2007b}. It has been shown, for the binary case, that, based on such a maximum acyclic agreement forest $\cF$, there exists a minimum hybridization network whose number of provided reticulation events is one less than the number of components in $\cF$ \cite{Baroni2005}. While, given two rooted binary phylogenetic trees, there exist some algorithms solving the maximum acyclic agreement forests as well as the minimum hybridization network problem, the nonbinary variant of both problems has attracted only less attention. More precisely, although there exist some approximation algorithms as well as exact algorithms \cite{Iersel2014,Whidden2014} computing maximum nonbinary agreement forests, until now there does not exist an algorithm computing nonbinary minimum hybridization networks.

The algorithm presented in this work was developed in respect to the algorithm \textsc{allHNetworks} \cite{Albrecht2014,Albrecht2015} computing a relevant set of minimum hybridization networks for multiple binary phylogenetic trees on the same set of taxa. Broadly speaking, those networks are computed by inserting each of the input trees step wise to a so far computed network which is done, basically, in three major steps. In a first step, all embedded trees of a so far computed network $N$ are extracted by selecting exactly one in-edge of each reticulation node. Notice that, at the beginning, $N$ is initialized with the first input tree $T_1$ containing only one embedded tree. The more reticulation edges exist in a network, the more embedded trees can be extracted. More precisely, given $r$ reticulation nodes of in-degree $2$, there may exist up to $2^r$ different embedded trees. Now, for each of those embedded trees $T'$, in a second step, all maximum acyclic agreement forests corresponding to $T'$ and the current input tree $T_i$ are computed which is done by applying the algorithm \textsc{allMAAFs} \cite{Scornavacca2012}. Finally, each component (except the root component $F_{\rho}$) of such a maximum acyclic agreement forest is inserted into the so far computed network in a certain way. 

In this article, we present the first non-naive algorithm --- called \textsc{allMulMAAFs} --- calculating a relevant set of nonbinary maximum acyclic agreement forests for two rooted (nonbinary) phylogenetic trees on the same set of taxa, which is a first step to adapt the algorithm \textsc{allHNetworks} to trees containing nodes providing more than two outgoing edges. Our paper is organized as follows. We first introduce some basic definitions. Next, the algorithm \textsc{allMulMAFs} is presented that is extended in a subsequent section to the algorithm \textsc{allMulMAAFs} by introducing a tool through which agreement forests can be turned into acyclic agreement forests. 

\section{Preliminaries}

In this section, we introduce some basic definitions referring to phylogenetic trees, hybridization networks, and nonbinary agreement forests based on the work of Huson \textit{et al.} \cite{Huson2007}, Scornavacca \textit{et al.} \cite{Scornavacca2012}, and van Iersel \textit{et al.} \cite{Iersel2014}. We therefor assume that the reader is familiar with basic graph-theoretic concepts.\\

\textbf{Phylogenetic trees.} A rooted phylogenetic $\cX$-tree $T$ is a tree whose edges are directed from the root to the leaves and whose nodes, except for the root, have a degree not equal to $2$. There exists exactly one node of in-degree $0$, namely the root. Each inner node has in-degree $1$ and out-degree larger or equal to $2$ whereas each leaf has in-degree $1$ and out-degree $0$. If $T$ is a \emph{bifurcating} or \emph{binary} tree, its root has in-degree $0$ and out-degree $2$, each inner node an in-degree of $1$ and an out-degree of $2$, and each leaf an in-degree of $1$ and out-degree $0$. Notice that, in order to emphasize that a tree $T$ can contain inner nodes of in-degree larger than $2$, we say that $T$ is a \emph{multifurcating} or \emph{nonbinary} tree. The leaves of a rooted phylogenetic $\cX$-tree are bijectivley labeled by a taxa set $\cX$, which usually consists of certain species or genes and is denoted by $\cL(T)$. Considering a node $v$ of $T$, the label set $\cL(v)$ refers to each taxon that is contained in the subtree rooted at $v$. Additionally, given a set of phylogenetic trees $\cF$, the label set $\cL(\cF)$ denotes the union of each label set $\cL(F_i)$ of each tree $F_i$ in $\cF$.

Now, based on a taxa set $\cX' \subseteq \cX$, we can define a restricted subtree of a rooted phylogenetic $\cX$-tree, denoted by $T|_{\cX'}$. The restricted tree $T|_{\cX'}$ is computed by first repeatedly deleting each leaf that is either unlabeled or whose taxon is not contained in $\cX'$, resulting in a subgraph denoted by $T(\cX')$, and then by suppressing each node of both in- and out-degree $1$. 

A rooted nonbinary phylogenetic $\cX$-tree $T$ can contain \emph{mulitfurcating} or \emph{nonbinary} nodes, which are nodes of out-degree larger than or equal to $3$. We say a rooted phylogenetic $\cX$-tree $T'$ is a \emph{refinement} of $T$, if we can obtain $T$ from $T'$ by \emph{contracting} some of its edges. More precisely, an edge $e=(u,v)$, with $C_v$ being the set of children of $v$, is contracted by first deleting $v$ together with all of its adjacent edges (including $e$) and then by reattaching each node $c_i$ in $C_v$ back to $u$ by a inserting a new edge $(u,c_i)$. Moreover, in this context we further say that $T'$ is a \emph{binary refinement of $T$}, if $T'$ is binary. Similarly, if $T'$ is a refinement of $T$, we can obtain $T'$ from $T$ by \emph{resolving} some of its multifurcating nodes in the following way (cf.~Fig.~\ref{fig-resNode}). Let $v$ be a multifurcating node and let $C_v=\{c_1,\dots,c_n\}$ be its set of children, then, we can resolve $v$ as follows. First, a new node $w$ is created, which is attached to $v$ by inserting a new edge $(v,w)$. Second, we select a subset $C_v'$ of $C_v$, with $1<|C_v'|<|C_v|$, and, finally, we prune each node $c_i$ of $C_v'$ from $v$ and reattach $c_i$ to $w$ by inserting a new edge $(w,c_i)$.

\begin{figure}[bt]
\centering
\includegraphics[scale = 1]{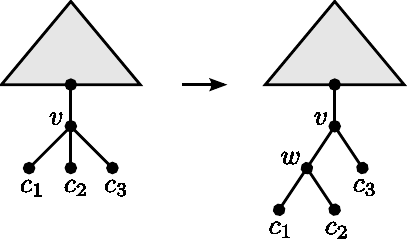}
\caption{Resolving a multifurcating node $v$ by re-attaching $c_1$ and $c_2$ to a new inserted node $w$.} 
\label{fig-resNode}
\end{figure}

Now, let $\cX'$ be a subset of the taxa set corresponding to a rooted phylogenetic $\cX$-tree $T$. Then, the \emph{lowest common ancestor} of $T$ corresponding to $\cX'$, shortly denoted by $\scLCA_T(\cX')$, is the farthest node $v$ from the root in $T$ with $\cX'\subseteq\cL(v)$. More precisely, $v$ is chosen such that $\cX'\subseteq\cL(v)$ holds and there does not exist a node $w$ with $\cX'\subseteq\cL(w)$ and $\cL(w)\subset\cL(v)$.\\

\textbf{Phylogenetic networks.} A \emph{rooted phylogenetic network} $N$ on $\cX$ is a rooted connected digraph whose edges are directed from the root to the leaves as defined in the following. There is exactly one node of in-degree $0$, namely the \emph{root}, and no nodes of both in- and out-degree $1$. The set of nodes of out-degree $0$ is called the \emph{leaf set of $N$} and is labeled one-to-one by the \emph{taxa set} $\cX$, also denoted by $\cL(N)$. In contrast to a phylogenetic tree, such a network may contain undirected but not any directed cycles. Consequently, $N$ can contain nodes of in-degree larger than or equal to $2$, which are called \emph{reticulation nodes} or \emph{hybridization nodes}. Moreover, each edge that is directed into such a reticulation node is called \emph{reticulation edge} \emph{hybridization edge}.\\

\textbf{Hybridization networks.} A hybridization network $N$ for a set of rooted nonbinary phylogenetic $\cX$-trees $\cT$ is a rooted phylogenetic network on $\cX$ \emph{displaying} a refinement $T_i'$ of each tree $T_i$ in $\cT$. More precisely, this means that for each tree $T_i$ in $\cT$ there exists a set $E_i'\subseteq E(N)$ of reticulation edges \emph{referring} to its refinement $T_i'$. This means, in particular, that we can obtain the tree $T_i'$ from $N$ by first deleting all reticulation edges that are not contained in $E_i'$ and then suppress all nodes of both in- and out-degree $1$. In this context, a reticulation edge (or hybridization edge) is an edge that is directed into a node with in-degree larger than or equal to $2$, which is denoted as reticulation node (or hybridization node). 

Given a hybridization network for a set of rooted nonbinary phylogenetic $\cX$-trees $\cT$, the \emph{reticulation number $r(N)$} is defined by 
\begin{equation}
r(N)=\sum_{v\in V:\delta^-(v)>0}(\delta^-(v)-1)=|E|-|V|+1,
\end{equation} 
where $V$ refers to the set of nodes of $N$ and $\delta^-(v)$ denotes the in-degree of a node $v$ in $V$. Moreover, based on the definition of the reticulation number, the (minimum or exact) hybridization number $h(\cT)$ for $\cT$ is defined by 
\begin{equation}
h(\cT)=\mbox{min}\{r(N):N\mbox{ displays a refinement of each }T_i\in\cT\}.
\end{equation}
A hybridization network displaying a set of rooted nonbinary phylogenetic $\cX$-trees $\cT$ with minimum hybridization number $h(\cT)$ is called a \emph{minimum hybridization network}. Notice that even in the simplest case, if $\cT$ consists only of two rooted binary phylogenetic $\cX$-trees, the problem of computing the hybridization number is known to be \emph{NP-hard} but fixed-parameter tractable \cite{Bordewich2007,Bordewich2007b}, which means that the problem is exponential in some parameter related to the problem itself, namely the hybridization number of $\cT$, but only at most polynomial in its input size, which is, in this context, the number of nodes and edges in $\cT$.\\

\textbf{Forests.} Let $T$ be a rooted nonbinary phylogenetic $\cX$-tree $T$. Then, we call any set of rooted nonbinary phylogenetic trees $\cF=\{F_1,\dots,F_k\}$ with $\cL(\cF)=\cX$ a \emph{forest on $\cX$}, if we have for each pair of trees $F_i$ and $F_j$ that $\cL(F_i)\cap\cL(F_j)=\emptyset$. Moreover, we say that $\cF$ is a \emph{forest for $T$}, if additionally for each component $F$ in $\cF$ the tree $F$ is a refinement of $T|_{\cL(F)}$. Lastly, given two forests $\cF$ and $\hat\cF$ for a rooted phylogenetic $\cX$-tree $T$, we say that $\hat\cF$ is a binary resolution of $\cF$, if for each component $\hat F$ in $\hat\cF$ there exists a component $F$ in $\cF$ such that $\hat F$ is a binary resolution of $F$. Lastly, let $\cF$ be a forest on a taxa set $\cX$, then by $\overline \cF$ we refer to the forest that is obtained from $\cF$ by deleting each element only consisting of an isolated node.\\\\

\textbf{Nonbinary agreement forests.} Given two rooted nonbinary phylogenetic $\cX$-trees $T_1$ and $T_2$. For technical purpose, we consider the root of both trees $T_1$ and $T_2$ as being a node that has been marked by new taxon $\rho\not\in \cX$. More precisely, let $r_i$ be the root of the tree $T_i$ with $i\in\{1,2\}$. Then, we first create a new node $v_i$ as well as a new leaf $\ell_i$ labeled by a new taxon $\rho\not\in\cX$ and then attach these nodes to $r_i$ by inserting the two edges $(v_i,r_i)$ and $(v_i,\ell_i)$ such that $v_i$ is the new root of the resulting tree. Now, an \emph{agreement forest} for two so marked trees $T_1$ and $T_2$ is a forest $\cF=\{F_{\rho},F_1,\dots,F_k\}$ on $\cX \cup \{\rho\}$ satisfying the following three conditions.

\begin{itemize}
\item[(1)] Each component $F_i$ with taxa set $\cX_i$ equals a refinement of $T_1|_{\cX_i}$ and $T_2|_{\cX_i}$, respectively. 
\item[(2)] There is exactly one component, denoted as $F_{\rho}$, containing $\rho$.
\item[(3)] Let $\cX_{\rho},\cX_1,\dots,\cX_k$ be the taxa sets corresponding to $F_{\rho},F_1,\dots,F_k$. All trees in $\{T_1(\cX_i)|i\in\{\rho,1,\dots,k\}\}$ and $\{T_2(\cX_i)|i\in\{\rho,1,\dots,k\}\}$ are edge disjoint subtrees of $T_1$ and $T_2$, respectively.
\end{itemize}

Let $\cF$ be an agreement forests for two rooted phylogenetic $\cX$-trees $T_1$ and $T_2$ and let $E$ be a set only consisting of edges in $\cF$. Then, by $\cF\ominus E$ we refer to the forest $\cF'$ that is obtained from $\cF$ by contracting each edge in $E$. Based on this definition, we say that $\cF$ is \emph{relevant}, if there does not exist an edge $e$ in $\cF$ such that $\cF\ominus \{e\}$ is still an agreement forest for $T_1$ and $T_2$. 

Moreover, a \emph{maximum agreement forest} for two rooted nonbinary phylogenetic $\cX$-trees $T_1$ and $T_2$ is an agreement forest of minimal size, which implies that there does not exist a smaller set of components fulfilling the properties of an agreement forest for $T_1$ and $T_2$ listed above. Additionally, we call an agreement forest $\cF$ for $T_1$ and $T_2$ \emph{acyclic}, if its underlying \emph{ancestor-descendant graph} $AG(T_1,T_2,\cF)$ does not contain any directed cycles (cf.~Fig.~\ref{fig-AG}). This directed graph contains one node corresponding to precisely one component of $\cF$ and an edge $(F_i,F_j)$ for a pair of its nodes $F_i$ and $F_j$, with $i\neq j$, if,
\begin{itemize}
\item[(i)] regarding $T_1$, there is a path leading from the root of $T_1(\cX_i)$ to the root of $T_1(\cX_j)$ containing at least one edge of $T_1(\cX_i)$, 
\item[(ii)] or, regarding $T_2$, there is a path leading from the root of $T_2(\cX_i)$ to the root of $T_2(\cX_j)$ containing at least one edge of $T_2(\cX_i)$.
\end{itemize}
In this context, $\cX_i\subseteq \cX$ and $\cX_j\subseteq \cX$ refers to the set of taxa that are contained in $F_i$ and $F_j$, respectively. Again, we call an acyclic agreement forest consisting of a minimum number of components a \emph{maximum acyclic agreement forest}. Notice that, in the binary case, for a maximum acyclic agreement forest containing $k$ components there exists a hybridization network whose reticulation number is $k-1$ \cite{Baroni2005}. This means, in particular, if a maximum acyclic agreement forest for two binary phylogenetic $\cX$-trees contains only one component, both trees are congruent.\\

\textbf{Acyclic orderings.} Now, if $\cF$ is acyclic and, thus, $\scAG(T_1,T_2,\cF)$ does not contain any directed cycles, one can compute an acyclic ordering $\Pi$ as already described in the work of Baroni \textit{et al.} \cite{Baroni2006}. First, select the node $v_{\rho}$ corresponding to $F_{\rho}$ of in-degree $0$ and remove $v_{\rho}$ together with all its incident edges. Next, again choose a node $v_1$ with in-degree $0$ and remove $v_1$. By continuing this way, until all nodes have been removed, one receives the ordering $\Pi=(v_{\rho},v_1,\dots,v_k)$. Notice that, since the graph does not contain any cycles, such an ordering always has to exist. In the following, we call the ordering of components corresponding to each node in $\Pi$, denoted by $(F_{\rho}, F_1,\dots,F_k)$, an acyclic ordering of $\cF$. As during each of those steps there can occur multiple nodes of in-degree $0$, especially, if $\cF$ contains multiple components consisting only of one taxon, such an acyclic ordering is in general not unique.\\

\begin{figure}[bt]
\centering
\includegraphics[scale = 1]{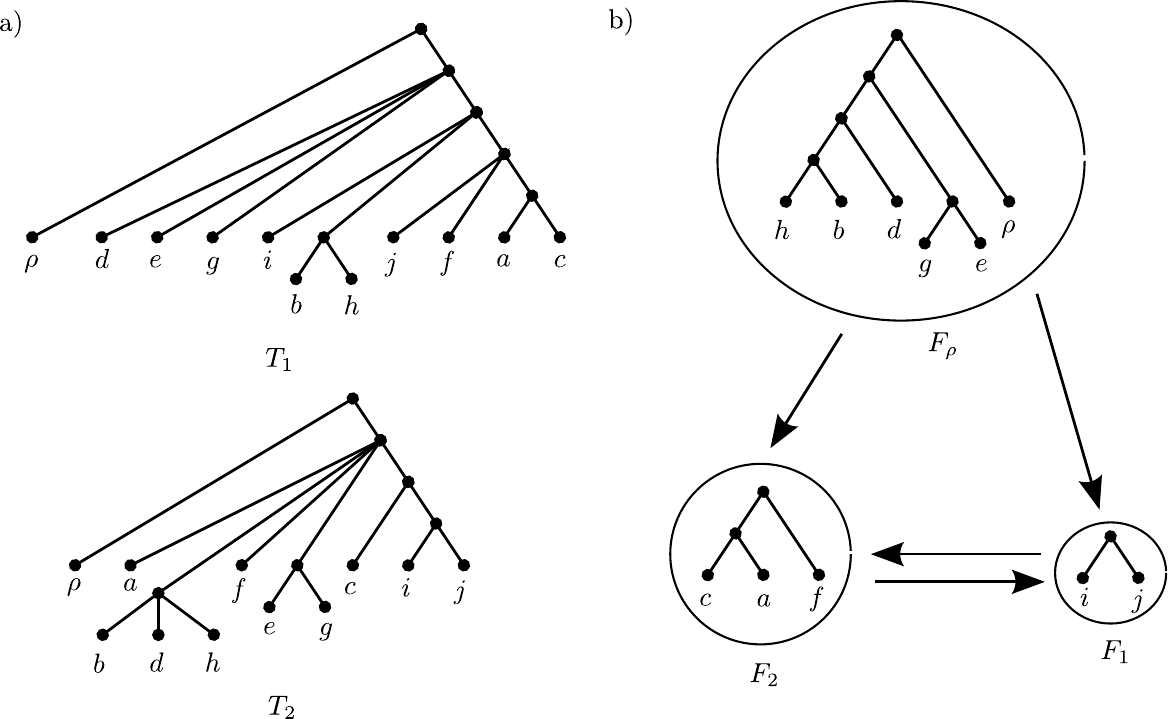}
\caption{(a) Two rooted nonbinary phylogenetic $\cX$-trees $T_1$ and $T_2$. (b) The ancestor-descendant graph $AG(T_1,T_2,\cF)$ with $\cF=\{F_{\rho},F_1,F_2\}$. Notice that the component corresponding to each node of the graph is drawn inside.} 
\label{fig-AG}
\end{figure}

\textbf{Trees reflecting agreement forests}. Let $\cF=\{F_{\rho},F_1,F_2,\dots,F_{k-1}\}$ be an agreement forest for two rooted (nonbinary) phylogenetic $\cX$-trees $T_1$ and $T_2$, then, a tree $T_i(\cF)$ for $i\in\{1,2\}$ corresponds to the tree $T_i$ reflecting each component in $\cF$ (cf.~Fig.~\ref{31-fig-TRef}). Generally speaking, in such a tree some of its nodes are resolved such that the definition of an agreement forest can be applied in terms of the resulting tree and $\cF$. Technically speaking, such a tree $\hat T_i=T_i(\cF)$ satisfies the following two properties.

\begin{itemize}
\item[(1)] Each component $F_j$ in $\cF$ refers to a restricted subtree of $\hat T_i|_{\cL(F_j)}$. 
\item[(2)] All trees in $\{\hat T_i(\cL(F_j))|j\in\{\rho,1,\dots,k-1\}\}$ are node disjoint subtrees in $\hat T_i$.
\end{itemize}

We can construct such a tree $T_i(\cF)$ by reattaching the components of $\cF$ back together in a specific way as follows. Let $\Pi_{\cF}=(F_0,F_1,\dots,F_k)$, with $F_0=F_{\rho}$, be an acyclic ordering that can be obtained from $AG(T_i,T_i,\cF)$ as discussed above. Notice that, as this graph is based only on one of both trees, this graph cannot contain any directed cycles and, thus, $\Pi_{\cF}$ always exists. Now, each of those components in $\Pi_{\cF}$, beginning with $F_1$, is added sequentially to a growing tree $T^*$ (initialized with $F_0$) as follows.

\begin{itemize}
\item[(i)] Let $\cX_{<m}$ be the union of each taxa set corresponding to each component $F_l$ in $\Pi_{\cF}$ with $l<m$, i.e., $\cX_{<m}=\bigcup^{m-1}_l\cL(F_l)$, and let $\cX_m$ be the taxa set corresponding to $F_m$. Moreover, let $P_m=(v_0^m,v_1^m,\dots,v_n^m)$ be those nodes lying on the path connecting the node $v_0^m=\scLCA_{T_i}(\cX_m)$ and the root $v_n^m$ of $T_i$ such that $v_q^m$, with $q\in\{1,\dots,n\}$, is the parent of $v_{q-1}^m$. Then, $$v'=\min_q\{v_q^m:v_q^m \in P \wedge \cL(v_q^m)\cap\cX_{<m}\ne\emptyset\}.$$
\item[(ii)] Let $\cX'$ be the set of taxa corresponding to the leaf set of $T_i(v')$ restricted to $\cX_{<m}$. Notice that, due to the definition of $v'$, this set $\cX'$ is not empty. Moreover, based on $\cX'$, let $v^*$ be the node in $T^*$ corresponding to $\scLCA_{T^*}(\cX')$.
\item[(iii)] Now, given $v^*$, the component $F_m$ is added to $T^*$ by connecting its root node $\rho_m$ to the in-edge of $v^*$. More precisely, first a new node $x$ is inserted into the in-edge of $v^*$ and then $\rho_m$ is connected to $x$ by inserting a new edge $(x,\rho_m)$.
\end{itemize}

\begin{figure}[bt]
\centering
\includegraphics[scale =1.3]{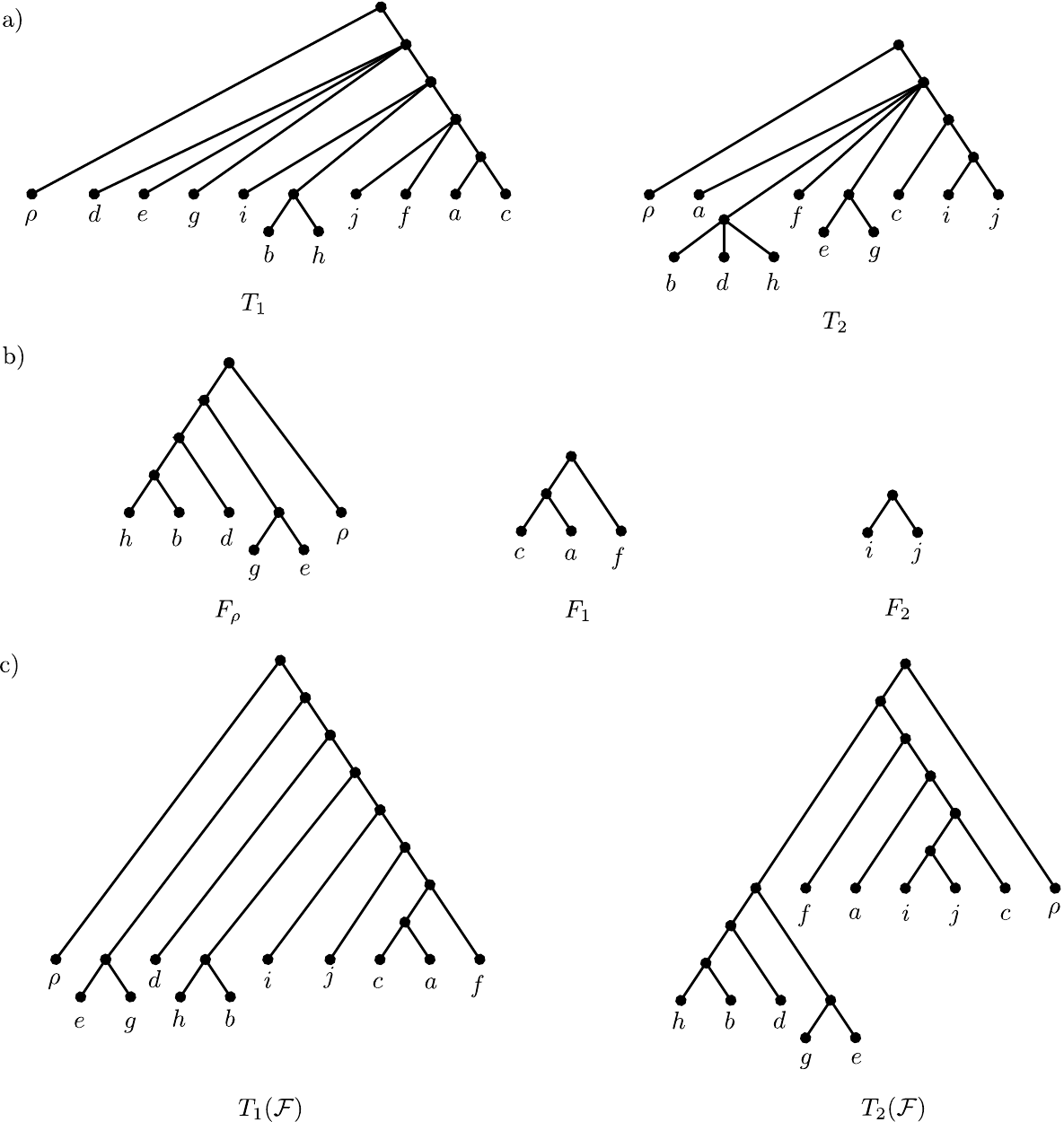}
\caption[Trees reflecting an agreement forest]{(a) Two rooted nonbinary phylogenetic $\cX$-trees $T_1$ and $T_2$. (b) An agreement forest $\cF$ for $T_1$ and $T_2$. (c) Two trees $T_1(\cF)$ and $T_2(\cF)$ both reflecting $\cF$ and being calculated in terms of the acyclic ordering $(F_{\rho},F_2,F_1)$ and $(F_{\rho},F_1,F_2)$, respectively.} 
\label{31-fig-TRef}
\end{figure}
Notice that, since there can exist multiple acyclic orderings for an acyclic agreement forest $\cF$, the tree $T_i(\cF)$ is in general not unique.

\clearpage
\section{The algorithm \textsc{allMulMAFs}}

In this section, we show how to modify the algorithm \textsc{allMAAFs} \cite{Scornavacca2012} such that the output consists of all relevant maximum agreement forests for two rooted \emph{nonbinary} phylogenetic $\cX$-trees. Similar to the algorithm \textsc{allMAAFs}, this algorithm is again based on processing common and contradicting cherries. In order to cope with nonbinary nodes, however, now for an internal node one has to consider more than one cherry and, before cutting a particular set of edges, one first has to resolve some nonbinary nodes. Furthermore, in respect to the definition of relevant maximum agreement forests, when expanding contracted nodes one has to take care on not generating any contractible edges.

In the following, we will first introduce some further notations necessary for describing the algorithm \textsc{allMulMAFs}. Moreover, we give a detailed formal proof establishing the correctness of the algorithm, which means, in particular, that we will show that the algorithm calculates all relevant maximum agreement forests for two rooted \emph{nonbinary} phylogenetic $\cX$-trees. Finally, we end this section by discussing its theoretical worst-case runtime.

\subsection{Notations}
\label{sec-def}

Before going into details, we have to give some further notations that are crucial for the following description of the algorithm.\\

\textbf{Removing leaves.} Given a rooted (nonbinary) phylogenetic $\cX$-tree $R$, a \emph{leaf $\ell$ is removed} by first deleting its in-edge and then by suppressing its parent $p$, if, after $\ell$ has been deleted, $p$ has out-degree $1$.\\

\textbf{Cherries.} Let $R$ be a rooted (nonbinary) phylogenetic $\cX$-tree and let $\ell_a$ and $\ell_c$ be two of its leaves that are adjacent to the same parent node $p$ and labeled by taxon $a$ and $c$, respectively. Then, we call the set consisting of the two taxa $\{a,c\}$ a \emph{cherry of $R$}, if the children of $p$ are all leaves. Now, let $\{a,c\}$ be a cherry of $R$ and let $\cF$ be a forest on a taxa set $\cX'$ such that $\overline\cF$ is a forest for $R$. Then, we say $\{a,c\}$ is a \emph{contradicting cherry of $R$ and $\cF$}, if $\cF$ does not contain a tree containing $\{a,c\}$. Otherwise, if such a tree exists in $\cF$, the cherry $\{a,c\}$ is called a \emph{common cherry of $R$ and $\cF$}.\\

\textbf{Contracting cherries.} Given a rooted (nonbinary) phylogenetic $\cX$-tree $R$, a cherry $\{a,c\}$ of $R$ can be contracted in two different ways (cf.~Fig.~\ref{32-fig-conCherry}). Either, if the two leaves $\ell_a$ and $\ell_c$ labeled by $a$ and $c$, respectively, are the only children of its parent node $p$, first the in-edge of both nodes is deleted and then the label of $p$ is set to $\{a,c\}$. Otherwise, if the two leaves $\ell_a$ and $\ell_c$ contain further siblings and, thus, its parent $p$ has an out-degree larger than $2$, just the in-edge of $\ell_a$ is deleted and the label of $\ell_c$ is replaced by $\{a,c\}$.\\

\begin{figure}[bt]
\centering
\includegraphics[scale = 1.3]{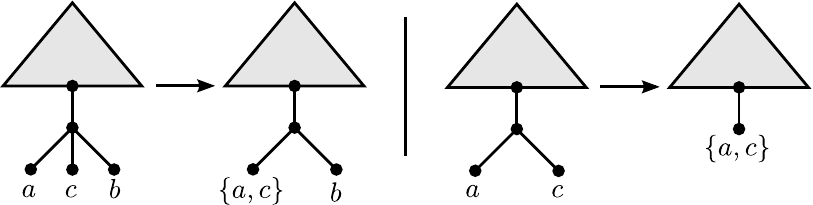}
\caption[Two different ways of contracting a cherry regarding the algorithm \textsc{allMulMAFs}]{Two different ways of contracting a cherry $\{a,c\}$ depending on its number of siblings which is one on the left hand side and zero on the right hand side and .} 
\label{32-fig-conCherry}
\end{figure}

\textbf{Cutting edges.} Given a rooted (nonbinary) phylogenetic $\cX$-tree $F$, the in-edge $e_v$ of a node $v$ is cut as follows (cf.~Fig.~\ref{32-fig-cutEdge}). First $e_v$ is deleted and then its parent node $p$ is suppressed, if, after the deletion of $e_v$, $p$ has out-degree $1$. Note that by cutting an edge in $F$, two rooted (nonbinary) phylogenetic trees with taxa set $\cX'$ and $\cX''$ are generated with $\cX'\cup\cX''=\cX$ and $\cX'\cap\cX''=\emptyset$.

Moreover, let $\cF$ be a set of rooted (nonbinary) phylogenetic $\cX$-trees and let $E$ be a set of edges in which each edge $e$ is part of a tree in $\cF$. Then, in order to ease reading, we write $\cF-E$ to denote the cutting of each edge $e\in E$ within its corresponding tree in $\cF$.\\

\begin{figure}[bt]
\centering
\includegraphics[scale = 1.5]{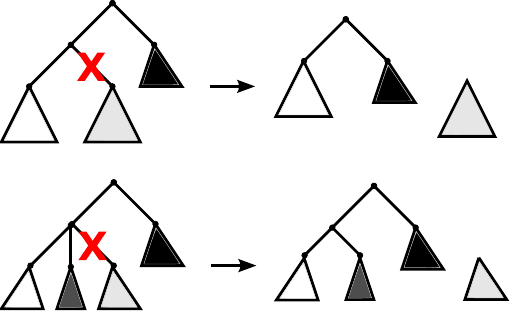}
\caption[Two different ways of cutting an edge regarding the algorithm \textsc{allMulMAFs}]{Two different ways of cutting an edge depending on the out-degree of its source node which is $2$ at the top and $3$ at the bottom.} 
\label{32-fig-cutEdge}
\end{figure}

\textbf{Pendant edges.} Given a rooted (nonbinary) phylogenetic $\cX$-tree $F$ and two leaves $\ell_a$ and $\ell_c$ labeled by taxon $a$ and $c$, respectively, that are not adjacent to the same node. Then, the set of \emph{pendant edges $E_B$ for $a$ and $c$} is based on a refinement of $F$, shortly denoted by $F[a\sim c]$, which is obtained from $F$ as follows. Let $(\ell_a,v_1,v_2,\dots,v_n,\ell_c)$ be the path connecting the two leafs $\ell_a$ and $\ell_c$ in $F$. Then, each node $v\in\{v_1,v_2,\dots,v_n\}$ with $\delta^+(v)>2$ is turned into a node of out-degree $2$ as follows. First a new node $w$ is created that is attached to $v$ by inserting a new edge $(v,w)$ and then each out-going edge $(v,x)$ with $x\not\in\{w,\ell_a,v_1,v_2,\dots,v_n,\ell_c\}$ is deleted followed by reattaching the node $x$ to $w$ by inserting a new edge $(w,x)$. Now, regarding $F[a\sim c]$, let $(\ell_a',v_1',v_2',\dots,v_n',\ell_c')$ be the path connecting the two nodes $\ell_a'$ and $\ell_c'$ labeled by $a$ and $c$, respectively. Then, $E_B$ consists of each edge $(u',v')$ with $u'\in\{v_1',v_2',\dots,v_n'\}$ and $v'\not\in\{\ell_a',v_1',v_2',\dots,v_n',\ell_c'\}$ (cf.~Fig.~\ref{32-fig-case_3b}).

\begin{figure}[bt]
\centering
\includegraphics[scale = 1.3]{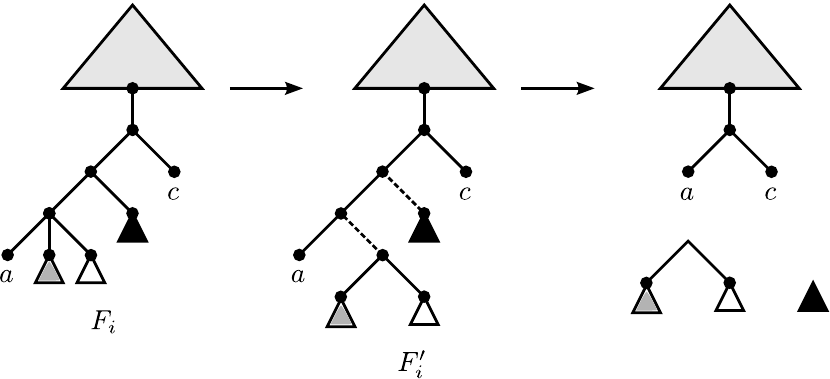}
\caption[An illustration of Case 3b following the description the algorithm \textsc{allMulMAFs}]{An illustration of Case 3b. First $F_i$ is refined into $F_i'=F_i[a\sim c]$ and then each pendant subtree of the path connecting both nodes labeled by $a$ and $c$, respectively, are cut. Note that each dashed edge of $F_i'$ is part of the pendant edge set $E_B$ for $a$ and $c$.}
\label{32-fig-case_3b}
\end{figure}

Moreover, given a forest $\cF$ on $\cX$ containing a tree $F$ with two leaves $\ell_a$ and $\ell_c$ labeled by taxon $a$ and $c$, respectively, that are not adjacent to the same node, then, by $\cF[a\sim c]$ we refer to $\cF$ in which $F$ is replaced by $F[a\sim c]$.\\ 

\textbf{Labeled nodes.} Let $R$ be a rooted (nonbinary) phylogenetic $\cX$-tree, then, by $\ell(R)$ we denote the number of its labeled nodes. Moreover, let $\cF$ be a forest on $\cX'$ such that $\overline\cF$ is a forest for $R$. Then, $\ell(\overline\cF)$ refers to the number of labeled nodes that are contained in each tree of $\overline\cF$. Additionally, we write $\ell(R)\equiv\ell(\overline\cF)$, if $\ell(R)$ equals $\ell(\overline\cF)$ and if for each labeled node $v_R$ in $R$ there exists a labeled node $v_F$ in $\overline\cF$ such that $\cL(v_R)=\cL(v_F)$. Moreover, if $\ell(R)\equiv\ell(\overline\cF)$ holds, we say that a leaf $\ell$ in $R$ \emph{refers} to a leaf $\ell'$ in $\cF$, if $\cL(\ell)$ equals $\cL(\ell')$.

\subsection{The algorithm}
\label{32-sec-algAF}

In this section, we give a description of the algorithm \textsc{allMulMAFs} calculating a particular set of nonbinary maximum agreement forests for two rooted phylogenetic $\cX$-trees $T_1$ and $T_2$. More specifically, as shown by an upcoming formal proof, this set consists of all relevant agreement forests for both trees. Before that, however, we want to give a remark emphasizing that the algorithm is based on a previous published algorithm that solves a similar problem dealing with rooted binary phylogenetic $\cX$-trees.

\begin{remark}
Our algorithm is an extension of the algorithm \textsc{allMAAFs} \cite{Scornavacca2012} computing all maximum \emph{acyclic} agreement forests for two rooted binary phylogenetic $\cX$-trees $T_1$ and $T_2$. Notice that the work of Scornavacca~\textit{et al.} \cite{Scornavacca2012} also contains a formal proof showing the correctness of the presented algorithm. The algorithm \textsc{allMulMAFs} presented here has a similar flavor and, thus, our notation basically follows the notation that has already been used for the description of the algorithm \textsc{allMAAFs}.
\end{remark}

Broadly speaking, given two rooted (nonbinary) phylogenetic $\cX$-trees $T_1$ and $T_2$ as well as a parameter $k\in\mathbb{N}$, our algorithm acts as follows. Based on the topology of the first tree $T_1$, the second tree $T_2$ is cut into several components until either the number of those components exceeds $k$ or the set of components fulfills each property of an agreement forest for $T_1$ and $T_2$. To ensure that there does not exist an agreement forest consisting of less than $k$ components, the following steps can be simply conducted by step-wise increasing parameter $k$ beginning with $k=0$. Thus, as far as our algorithm reports an agreement forest for $T_1$ and $T_2$ of size $k$, this agreement forest must be of minimum size and, hence, must be a maximum agreement forest for both input trees. In order to speed up computation, one can either set $k$ to a lower bound calculated by particular approximation algorithms as, for instance, given in van Iersel \textit{et al.} \cite{Iersel2014}, or directly to the hybridization number calculated by applying less complex algorithms, e.g., the algorithm \textsc{TerminusEst} \cite{Piovesan2014}. 

The algorithm \textsc{allMulMAFs} takes as input two rooted (nonbinary) phylogenetic $\cX$-trees $T_1$ and $T_2$ as well as a parameter $k\in\mathbb{N}$. If $k<h(T_1,T_2)$ holds, an empty set is returned. Otherwise, as we will show later in Section~\ref{32-sec-cor1}, if $k$ is larger than or equal to $h(T_1,T_2)$, the output $\boldsymbol\cF$ of \textsc{allMulMAFs} contains all relevant maximum agreement forests for $T_1$ and $T_2$. Throughout the algorithm three specific tree operations are performed on both input trees. Either a leaf is removed, subtrees are cut, or a common cherry is contracted. 

The algorithm \textsc{allMulMAFs} contains a recursive subroutine, in which the input of each recursion consists of a rooted (nonbinary) phylogenetic $\cX$-tree $R$, a forest $\cF$ on some taxa set $\cX'$ with $\overline\cF$ being a forest for $R$, a parameter $k\in\mathbb{N}$, and a map $M$. This map $M$ is necessary for undoing each cherry reduction that has been applied to each component of the resulting forest. For that purpose, $M$ maps a set of taxa $\tilde\cX$ to a triplet $(\cX_1,\cX_2,B)$ with $\cX_1\cup\cX_2=\tilde\cX$, $\cX_1\cap\cX_2=\emptyset$, and $B\in\{\top,\bot\}$, where $B$ denotes the way of how a cherry is expanded (as discussed below). In order to ease reading, by $M[\tilde\cX]\leftarrow B$ we refer to the operation on $M$ mapping $\tilde\cX$ to $B$. This means, in particular, if $M$ already contains an element with taxa set $\tilde\cX$ this element is replaced. For instance, if $M[\tilde\cX]=(\cX_1,\cX_2,\top)$, by $M[\tilde\cX]\leftarrow \bot$ the taxa set $\tilde\cX$ is remapped to $(\cX_1,\cX_2,\bot)$ so that after this operation $M[\tilde\cX]=(\cX_1,\cX_2,\bot)$. \\

\textbf{Expanding agreement forests.} The expansion of an agreement forest $\cF$ is done by applying the following steps to $\cF$. Choose a leaf $v$ corresponding to a component $F_i$ in $\cF$ whose taxon $\cL(v)$ is contained in $M$. Let $(\cX_1,\cX_2,B)$ be the triplet referring to $M[\cL(v)]$, then, depending on $B\in\{\top,\bot\}$, one of the following two operations is performed as illustrated in Figure~\ref{32-fig-expCher}.
\begin{itemize}
\item If $B$ equals $\bot$, replace $v$ in $F_i$ by first creating two new nodes $w_1$ and $w_2$ and then by labeling $w_1$ and $w_2$ by $\cX_1$ and $\cX_2$, respectively. Finally, both nodes $w_1$ and $w_2$ are attached to $v$ by inserting a new edge $(v,w_1)$ and $(v,w_2)$ 
\item Otherwise, if $B$ equals $\top$, replace $v$ in $F_i$ by first creating a new node $w$ and then by labeling $w$ and $v$ by $\cX_1$ and $\cX_2$, respectively. Finally, $w$ is attached to the parent $p$ of $v$ by inserting a new edge $(p,w)$.
\end{itemize}
These steps are repeated in an exhaustive way until each taxon of each leaf in $\cF$ is not contained in $M$. As a result, the expanded forest corresponds to a nonbinary agreement forest for the two input trees $T_1$ and $T_2$. Notice that in the following, by saying a cherry is expanded \emph{in respect of~$\bot$} or \emph{in respect of~$\top$}, we refer to one of both ways as described above.\\

\begin{figure}[bt]
\centering
\includegraphics[scale = 1.3]{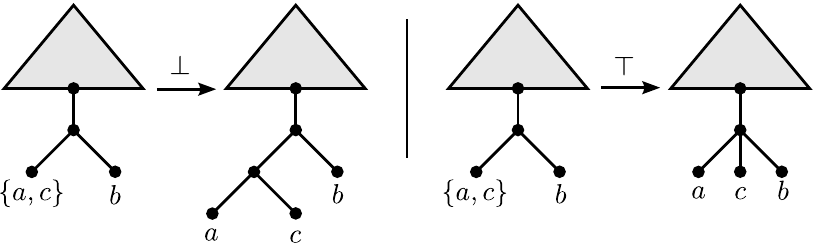}
\caption[Two different ways of expanding a cherry regarding the algorithm \textsc{allMulMAFs}]{Expanding a cherry $\{a,c\}$ in respect of $\bot$ (left) and in respect of $\top$ (right).} 
\label{32-fig-expCher}
\end{figure}

In the following, a description of the recursive algorithm \textsc{allMulMAFs} is given. Here we assume that, at the beginning, $R$ is initialized by $T_1$, $\cF$ by $\{T_2\}$ and $M$ by $\emptyset$, with $T_1$ and $T_2$ being two rooted (nonbinary) phylogenetic $\cX$-trees. Then, during each recursive call, $\overline\cF$ is a forest for $R$ with $\ell(R)\equiv\ell(\overline\cF)$ and, depending on the size of $\cF$ (cf.~Case~1a--c) and the choice of the next cherry that is selected from $R$ (cf.~Case~2~and~3), the following steps are performed.\\

\textbf{Case 1a.} If $\cF$ contains more than $k$ components, the computational path is aborted immediately and the empty set is returned.\\

\textbf{Case 1b.} If $R$ only consists of a single leaf, each $F_i$ in $\cF$ is expanded as prescribed in~$M$, and, finally, returned.\\

\textbf{Case 1c.} If there exists a specific leaf $\ell$ in $R$ that refers to an isolated node in $\cF$, this leaf $\ell$ is removed from $R$ resulting in $R'$. Next, the algorithm branches into a new path by recursively calling the algorithm with $R'$, $\cF$, $k$, and $M'$ corresponding to $M$ where $\cL(\ell)$ is re-mapped to $M[\cL(\ell)]\leftarrow\bot$. 

Otherwise, if such a leaf $\ell$ does not exist continue with Case~2.\\

\textbf{Case 2.} If there exists a common cherry $\{a,c\}$ of $R$ and $\cF$, the cherry $\{a,c\}$ is contracted in $R$ and $\cF$ resulting in $R'$ and $\cF'$. Second, the algorithm branches into a new path by recursively calling the algorithm with $R'$, $\cF'$, $k$, and $M'$, where $M'$ corresponds to $M$ that has been updated as follows. If both parents of $\{a,c\}$ in $R$ and $\cF$ have out-degree $\ge 3$, $\cL(a)\cup\cL(c)$ is mapped to $(\cL(a),\cL(c),\top)$, otherwise, $\cL(a)\cup\cL(c)$ is mapped to $(\cL(a),\cL(c),\bot)$. 

Otherwise, if such a common cherry does not exist, continue with Case~3.\\

\textbf{Case 3.} If there does not exist a common cherry of $R$ and $\cF$, a node $v$ in $R$ whose children are all leaves is selected. Now, \emph{for each} cherry $\{a,c\}$ of $v$, depending on the location of the leaves referring to $a$ and $c$ in $\cF$, one of the following two cases is performed.\\

\textbf{Case 3a.} If $a\not\sim_{\cF} c$ holds, and, thus, the leaves referring to $a$ and $c$ in $\cF$ are located in two different components, the algorithm branches into two computational paths by recursively calling the algorithm by $R$, $\cF-\{e_a\}$, $k$, and $M'$ as well as $R$, $\cF-\{e_c\}$, $k$, and $M'$, where $e_a$ and $e_c$ correspond to the in-edge of the leaf of $\cF$ referring to $a$ and $c$, respectively, and $M'$ is obtained from $M$ as follows. Let $p$ be the parent in $\cF$ of the leaf referring to $a$ (resp. $c$). If $p$ has out-degree larger than~$2$, nothing is done. Otherwise, if $p$ has out-degree~$2$, let $\ell$ be the sibling of the leaf labeled by $a$ (resp. $c$). Then, if $\ell$ is a leaf $M$ is updated so that $M'=M[\cL(\ell)]\leftarrow\bot$.\\

\textbf{Case 3b.} If $a\sim_{\cF} c$ holds, and, thus, in $\cF$ both leaves $\ell_a$ and $\ell_c$ referring to $a$ and $c$, respectively, are located in the same component $F_i$, the algorithm branches into the following three computational paths. First, similar to Case~3a, the algorithm is called by $R$, $\cF-\{e_a\}$, $k$, and $M'$ as well as $R$, $\cF-\{e_c\}$ $k$, and $M'$. Second, a third computational path is initiated by calling the algorithm with $R$, $\cF[a\sim c]-E_B$, $k$, and $M''$, where $E_B$ refers to the set of pendant edges in $\cF[a\sim c]$ and $M''$ is obtained from $M$ as follows.

Let $(\ell_a,v_1,\dots,v_n,\ell_c)$ denote the path connecting $\ell_a$ and $\ell_c$ in $\cF$. Then, $M''$ is obtained by updating $M$ as follows. If $v_1$ does not correspond to $\scLCA_{F_i}(\{a,c\})$, $\cL(\ell_a)$ is remapped to $M[\cL(\ell_a)]\leftarrow\bot$. Similarly, if $v_n$ does not correspond to $\scLCA_{F_i}(\{a,c\})$, $\cL(\ell_c)$ is remapped $M[\cL(\ell_c)]\leftarrow\bot$.

An illustration of this case is given in Figure~\ref{32-fig-case_3b}\\

We end the description of the algorithm by noting that the algorithm \textsc{allMulMAFs} always terminates, since during each recursive call either the size of $R$ decreases or the number of components in $\cF$ increases. More precisely, the size of $R$ is decreased by one either by deleting one of its leaves $\ell$ referring to an isolated node in $\cF$ (cf.~Case~1c) or by contracting a common cherry of $R$ and $\cF$ (cf.~Case~2). If $R$ is not decreased, at least one edge in $\cF$ is cut (cf.~Case~3) and, thus, its size increases at least by one. As each computational path of the algorithm stops if $R$ only consists of an isolated node or if $k$ edges have been cut, each recursive call does always make progress towards one of both abort criteria.

\begin{figure}[bt]
\centering
\includegraphics[scale = 1.3]{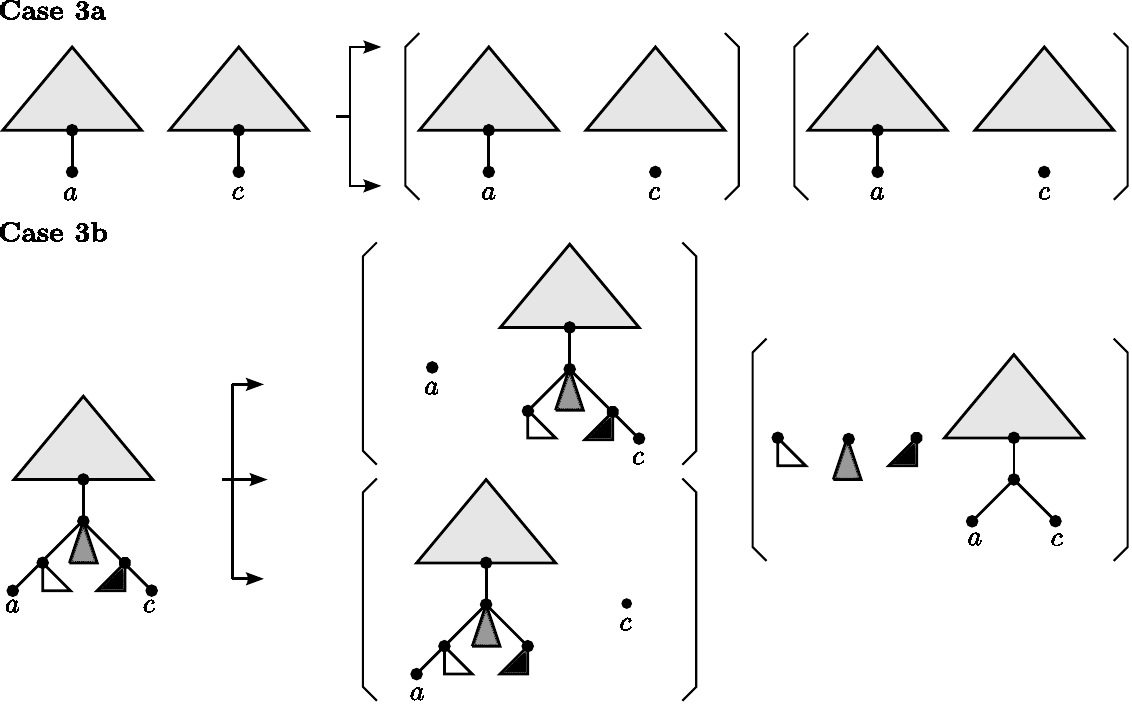}
\caption[An illustration of Case 3a following the description the algorithm \textsc{allMulMAFs}]{An illustration of Case~3a branching into two computational paths and Case 3b branching into three computational paths. Regarding Case~3b, one additional computational path is created in which all pendant subtrees lying on the path connecting the two nodes labeled by $a$ and $c$, respectively, are cut. } 
\label{fig-case_3ab}
\end{figure}

\subsection{Correctness of \textsc{allMulMAFs}}
\label{32-sec-cor1}

In this section, we establish the correctness of the algorithm \textsc{allMulMAFs}. However, before doing so, we want to give an important remark emphasizing the relation of our algorithm \textsc{allMulMAFs} to the previously presented algorithm \scAMa~\cite{Albrecht2015b}, which is a modification of the algorithm \textsc{allMAAFs} \cite{Scornavacca2012} improving the processing of contradicting cherries.

\begin{remark}
Given two binary phylogenetic $\cX$-trees, the algorithm \textsc{allMulMAFs} processes an ordered set of cherries in the same way as the algorithm \scAMa~omitting its acyclic check (henceforth denoted as \scAMSa) testing an agreement forests for acyclicity. This means, in particular, that our algorithm \textsc{allMulMAFs} is simply an extension of the algorithm \scAMSa~that is now able to handle nonbinary trees, but for binary trees still acts in the same way.
\label{32-rem-maafs}
\end{remark} 

As a consequence of Remark~\ref{32-rem-maafs}, the upcoming proof showing the correctness of \textsc{allMulMAFs} refers to the correctness of \scAMa~calculating all maximum acyclic agreement forests for two rooted binary phylogenetic $\cX$-trees \cite[Theorem~2]{Albrecht2015b}. In a first step, however, in order to ease the understanding of our proof, we will introduce a connective element between both algorithms, which is a modified version of our original algorithm --- called \textsc{allMulMAFs*} --- processing types of cherries that are not considered by computational paths corresponding to \textsc{allMulMAFs}.

\subsubsection{The algorithm \textsc{allMulMAFs*}}

Before describing the algorithm, we have to add further definitions that are crucial for what follows.\\

\textbf{Proper leaves.} Given a leaf $\ell$ of a rooted (nonbinary) phylogenetic $\cX$-tree $R$ labeled with taxon $a$ as well as a forest $\cF$ on some taxa set $\cX'$ such that $\overline\cF$ is a forest for $R$, $\ell$ is called a \emph{proper leaf of} $R$ and $\cF$, if the corresponding leaf in $\cF$ labeled by taxon $a$ is a child of some root.\\

\textbf{Pseudo cherries.} Given a rooted (nonbinary) phylogenetic $\cX$-tree $R$ as well as a forest $\cF$ on some taxa set $\cX'$ such that $\overline\cF$ is a forest for $R$, we call a set of two taxa $\{a,c\}$ a \emph{pseudo cherry for $R$ and $\cF$}, if the following two properties hold. First for each child $v$ of $\scLCA_R(\{a,c\})$ its leaf set $\cL(T(v))$ of size $n$ contains at least $n-1$ proper leaves. Second, the path connecting the two leaves in $R$ labeled by $a$ and $c$ contains at least one pendant proper leaf.\\

\textbf{Preparing cherries.} Given a rooted (nonbinary) phylogenetic $\cX$-tree $R$, a forest $\cF$ on some taxa set $\cX'$ such that $\overline\cF$ is a forest for $R$ as well as a cherry $\{a,c\}$, then, \emph{$\{a,c\}$ is prepared} as follows. If $\{a,c\}$ is not a pseudo cherry for $R$ and $\cF$, nothing is done. Else, the following two steps are conducted. First, each pendant proper leaf $\{\ell_1,\dots,\ell_n\}$ in $R$ lying on the path connecting the two leaves labeled by $a$ and $c$ is removed. Second, the two nodes in $\cF$ labeled by $a$ and $c$ are cut. Notice that, after preparing $\{a,c\}$, the node $\scLCA_R(\{a,c\})$ in $R$ is the parent of the two leaves labeled by taxon $a$ and $c$ and each component in $\cF$ referring to $\{\ell_1,\dots,\ell_n\}$ only consists of a single isolated node (cf.~Fig.~\ref{32-fig-pseudo}).\\ 

\begin{figure}[bt]
\centering
\includegraphics[scale = 1.3]{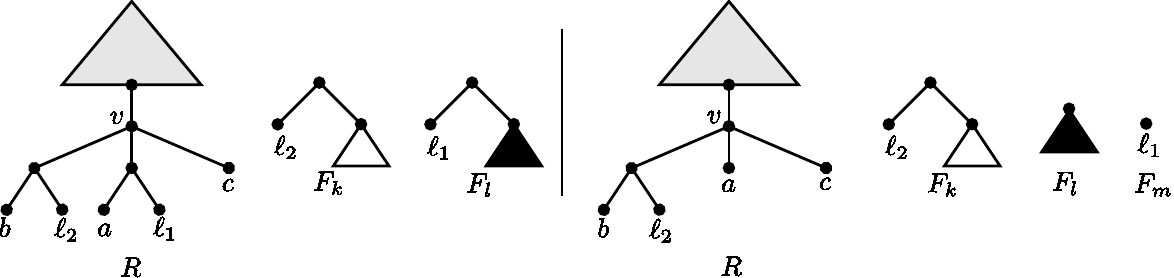}
\caption[An illustration of a pseudo cherry regarding the algorithm \textsc{allMulMAFs*}]{(left) An illustration of a pseudo cherry $\{a,c\}$. (right) The result of preparing the pseudo cherry $\{a,c\}$ given on the left hand side.} 
\label{32-fig-pseudo}
\end{figure}

Now, similar to the original algorithm, the algorithm \textsc{allMulMAFs*} is called by the same four parameters $R$, $\cF$, $M$ and $k$. Given two rooted phylogenetic $\cX$-trees $T_1$ and $T_2$, $R$ is initialized by $T_1$, $\cF$ by $\{T_2\}$ and $M$ by $\emptyset$. Depending on the size of $\cF$ (cf.~Case~1a--c) and the choice of the next cherry that is selected from $R$ (cf.~Case~2), the following steps are performed.\\

\textbf{Case 1a.} If $F$ contains more than $k$ components, the computational path is aborted immediately and an empty set is returned.\\

\textbf{Case 1b.} If $R$ only consists of a single leaf, each $F_i$ in $\cF$ is expanded with the help of $M$, and, finally, returned.\\

\textbf{Case 1c.} If there exists a specific leaf $\ell$ in $R$ that refers to an isolated node in $F$, this leaf $\ell$ is removed from $R$ resulting in $R'$. Next, the algorithm branches into a new path by recursively calling the algorithm with $R'$, $\cF$, $k$, and $M'$ corresponding to $M$ updated by $M[\cL(\ell)]\leftarrow\bot$. 

Otherwise, if such a leaf $\ell$ does not exist continue with Case~2.\\

\textbf{Case 2.} Select a subtree in $R$ in which each pair of taxa either represents a cherry or a pseudo cherry. Now, \emph{for each} (pseudo) cherry $\{a,c\}$, depending on the location of the leaves referring to $a$ and $c$ in $\cF$, one of the following three cases is performed. In a first step, however, the chosen cherry $\{a,c\}$ is prepared as described above. Moreover, $M$ is updated by $M[\cX_i]\leftarrow\bot$, where $\cX_i$ denotes the taxa set of each proper leaf that has been cut during the preceding preparation step.\\

\textbf{Case 2a.} If $\{a,c\}$ is a common cherry, $\{a,c\}$ is processed as described in Case~2 corresponding to the original algorithm \textsc{allMulMAFs}.\\

\textbf{Case 2b.} If $a\not\sim_{\cF} c$ holds, and, thus, the leaves referring to $a$ and $c$ in $\cF$ are located in two different components, $\{a,c\}$ is processed as described in Case~3a corresponding to the original algorithm \textsc{allMulMAFs}.\\

\textbf{Case 2c.} If $a\sim_{\cF} c$ holds, and, thus, the leaves referring to $a$ and $c$ in $\cF$ can be found in the same component $F_i$, $\{a,c\}$ is processed as described in Case~3b corresponding to the original algorithm \textsc{allMulMAFs}.\\

Notice that there are two main differences between the algorithm \textsc{allMulMAFs*} and the original algorithm \textsc{allMulMAFs}. First, a computational path corresponding to \textsc{allMulMAFs*} can additionally process pseudo cherries. Second, if during a recursive call $R$ contains a common cherry $\{a,c\}$ as well as a contradicting cherry $\{a,b\}$, \textsc{allMulMAFs*} additionally branches into a computational path processing $\{a,b\}$. In the following, we will call such a cherry $\{a,b\}$ a \emph{needless cherry} as we will show later that it can be neglected for the computation of maximum agreement forests.

\subsubsection{The algorithm \textsc{ProcessCherries}}

Lastly, we present a simplified version of the algorithm \textsc{allMulMAFs*} --- called \textsc{ProcessCherries} --- mimicking one of its computational by a \emph{cherry list} $\fcL=(\fc_1,\fc_2,\dots,\fc_n)$, in which each of its elements $\fc_i$ denotes a \emph{cherry action}. Such a cherry action $\fc_i=(\{a,c\},\phi_i)$ is a tuple that contains a (pseudo) cherry $\{a,c\}$ of the corresponding rooted phylogenetic $\cX$-tree $R_i$ and the forest $\cF_i$ as well as a variable $\phi_i$ $\in\{\cup_{ac},\nmid_a,\nmid_c,\cap_{ac}\}$ denoting the way $\{a,c\}$ is processed in iteration $i$. More precisely,

\begin{itemize}
\item $\cup_{ac}$ refers to contracting the cherry $\{a,c\}$ following Case~2a of the algorithm \textsc{allMulMAFs*}.
\item $\nmid_a$ and $\nmid_c$ refers to cutting taxon $a$ and $c$, respectively, of the cherry $\{a,c\}$ following Case~2b of the algorithm \textsc{allMulMAFs*}.
\item $\cap_{ac}$ refers to cutting each pendant subtree connecting taxon $a$ and taxon $c$ in $\cF_i$ following Case~2c of the algorithm \textsc{allMulMAFs*}.
\end{itemize}

\begin{algorithm}[bt]
\scriptsize
$M\leftarrow\emptyset$\;
\For{$i = 1,\dots,n$}{
	$(\{a,c\},\phi_i)\leftarrow\wedge_i$\;
	\If{$\{a,c\}$ is a cherry of $R$ or a pseudo cherry for $R$ and $\cF$}{
		\If{$\{a,c\}$ is a pseudo cherry for $R$ and $\cF$}{
			$(R,\cF,M)\leftarrow$ prepare pseudo cherry $\{a,c\}$\;
		}
		\If{$\{a,c\}$ is a common cherry of $R$ and $\cF$ and $\phi_i==\cup_{ac}$}{
			$(R,\cF,M)\leftarrow$ contract cherry $\{a,c\}$\;
		}
		\uElseIf{$\{a,c\}$ is a contradicting cherry of $R$ and $\cF$ and $\phi_i==\nmid_a$}{
			$e_a\leftarrow$ in-edge of node labeled by $a$ in $\cF$\;
			$(R,\cF,M)\leftarrow$ cut edge $e_a$ in $\cF$\;
		}
		\uElseIf{$\{a,c\}$ is a contradicting cherry of $R$ and $\cF$ and $\phi_i==\nmid_c$}{
			$e_c\leftarrow$ in-edge of node labeled by $c$ in $\cF$\;
			$(R,\cF,M)\leftarrow$ cut edge $e_c$ in $\cF$\;
		}
		\ElseIf{$\{a,c\}$ is a contradicting cherry of $R$ and $\cF$}{
			$\cF \leftarrow \cF[a\sim c]$\;
			$E_B\leftarrow$ set of pendant edges for $a$ and $c$ in $\cF$\;
			$(R,\cF,M)\leftarrow$ cut each edge in $E_B$\;
		}
		\Else{
			\Return($\emptyset$)\;
		}
		$(R,\cF,M)\leftarrow$ from $R$ remove each leaf referring to an isolated node in~$\cF$\;
	}
	\Else{
		\Return($\emptyset$)\;
	}
}
$\cF \leftarrow$ expand $\cF$ as prescribed in $M$\;
\Return($\cF$)\;
\caption{$\textsc{ProcessCherries}(R,\cF,(\wedge_1,\dots,\wedge_n))$}
\label{32-alg-proc}
\end{algorithm}

Now, given a cherry list $\fcL$, we say that $\fcL$ is a \emph{cherry list for $T_1$ and $T_2$}, if in each iteration $i$ the cherry $\{a,c\}$ of $\fc_i=(\{a,c\},\phi_i)$ is either contained in $R_i$ or $\{a,c\}$ is a pseudo cherry for $R_i$ and $F_i$. Moreover, after having prepared the cherry $\{a,c\}$, one of the following two conditions has to be satisfied.
\begin{itemize}
\item Either $\{a,c\}$ is a common cherry of $R_i$ and $F_i$ and $\phi_i=\cup_{ac}$,
\item or $\{a,c\}$ is a contradicting cherry of $R_i$ and $F_i$.
\end{itemize}
Notice that this is the case, if and only if calling $\textsc{ProcessCherries}(T_1,\{T_2\},\fcL)$ does not return the empty set (cf.~Alg.~\ref{32-alg-proc}).

\begin{figure}
\centering
\includegraphics[scale =1]{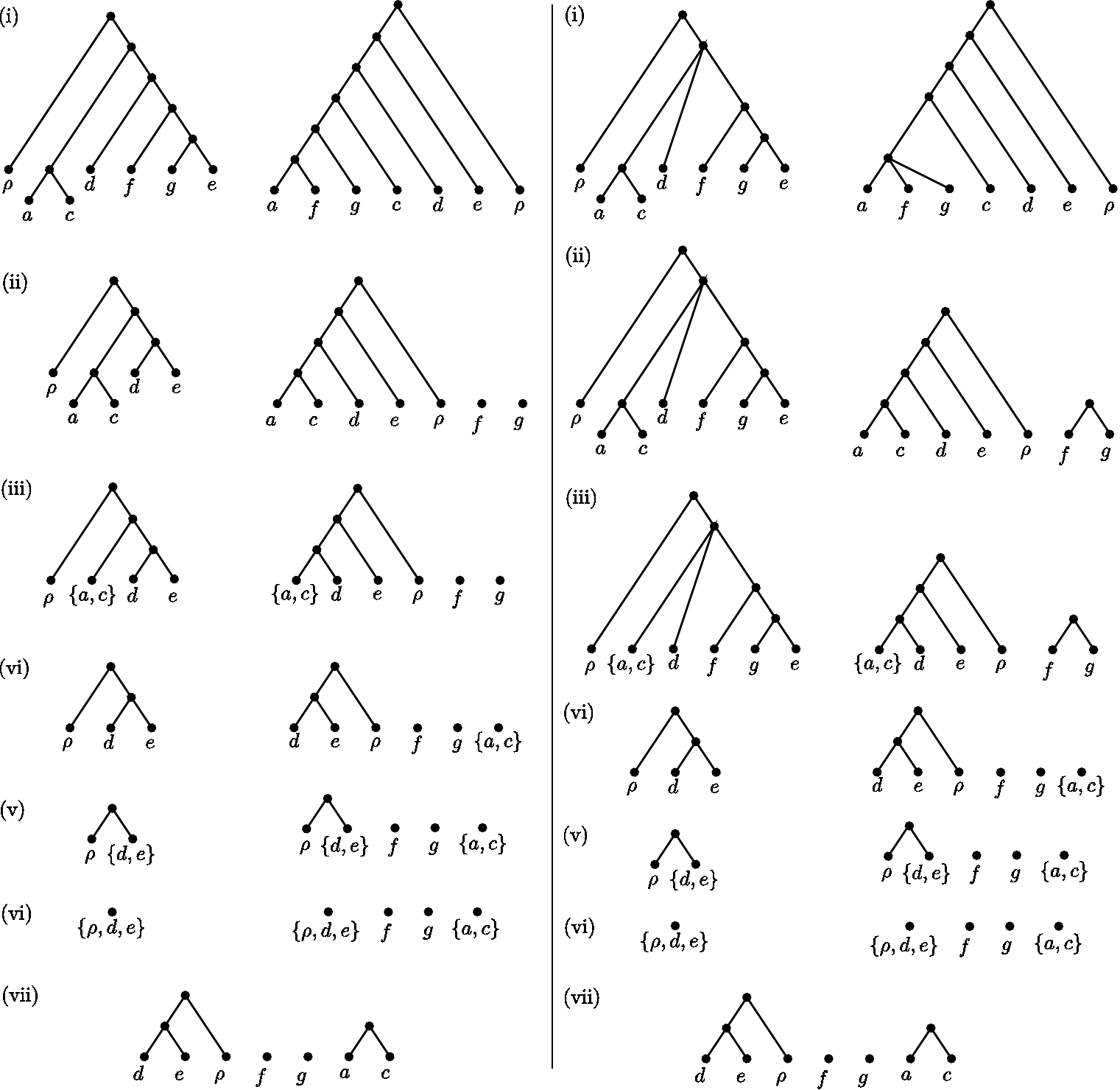}
\caption[Two examples of calling the algorithm \textsc{ProcessCherries}]{Two examples of calling \textsc{ProcessCherries} for two binary and nonbinary trees according to the cherry actions 
$(\{a,c\},\cap_{ac}),(\{a,c\},\cup_{ac}),(\{d,e\},\cap_{de}),(\{d,e\},\cup_{de}),$ and $(\{\rho,\{d,e\}\},\cup_{\rho de})$ conducted in sequential order. Notice that, as the two binary trees are binary resolutions of the two nonbinary trees, the resulting forest on the left hand side is a binary resolution of the resulting forest on the right hand side. Moreover, regarding Step (iii) on the right hand side, notice that the chosen cherry $\{d,e\}$ is a pseudo cherry and, thus, in Step (iv) the two components consisting of the single nodes labeled by $f$ and $g$, respectively, arise from preparing $\{d,e\}$.} 
\label{32-fig-procCher}
\end{figure}

\subsubsection{Proof of Correctness}

In this section, we will establish the correctness of the algorithm \textsc{allMulMAFs} by establishing the following theorem.

\begin{theorem}
Given two rooted (nonbinary) phylogenetic $\cX$-trees $T_1$ and $T_2$, by calling $$\textsc{allMulMAFs}(T_1,\{T_2\},\emptyset,k)$$ all relevant maximum agreement forests for $T_1$ and $T_2$ are calculated, if and only if $k\ge h(T_1,T_2)$.
\label{32-th-cor1}
\end{theorem}

\begin{proof}
The proof of Theorem~\ref{32-th-cor1} is established in several substeps. First, given two rooted (nonbinary) phylogenetic $\cX$-trees $T_1$ and $T_2$, we will show that a binary resolution of each maximum agreement forest for $T_1$ and $T_2$ can be computed by applying the algorithm \scAMSa~to a binary resolution of $T_1$ and $T_2$, where, as already mentioned, \scAMSa~denotes a modification of the algorithm \scAMa~omitting the acyclic check. Next, we will show that for an agreement forest $\cF$ calculated by \textsc{allMulMAFs} there does not exist en edge $e$ such that $\cF\ominus\{e\}$ is still an agreement forest for $T_1$ and $T_2$, which directly implies that $\cF$ is relevant. Moreover, we will show that, if a cherry list $\fcL$ for two binary resolutions $\hat T_1$ and $\hat T_2$ of $T_1$ and $T_2$, respectively, computes a maximum agreement forest $\hat\cF$ for $\hat T_1$ and $\hat T_2$, $\fcL$ is mimicking a computational path of the algorithm \textsc{allMulMAFs*} calculating an agreement forest $\cF$ for $T_1$ and $T_2$ such that $\hat\cF$ is a binary resolution of $\cF$. Lastly, we will show that each maximum agreement forest $\cF$ computed by \textsc{allMulMAFs*} is also computed by \textsc{allMulMAFs}.

\begin{lemma}
\label{32-lem-cor0}
Given two rooted (nonbinary) phylogenetic $\cX$-trees $T_1$ and $T_2$, for each relevant maximum agreement forest $\cF$ for $T_1$ and $T_2$ of size $k$ there exists a binary resolution $\hat\cF$ of $\cF$ that is calculated by $$\text{\scAMSa$(\hat T_1,\hat T_2,\hat T_1,\hat T_2,k),$}$$ where $\hat T_1$ and $\hat T_2$ refers to binary resolutions of $T_1$ and $T_2$, respectively.
\end{lemma}

\begin{proof}
First notice that the algorithm \scAMa~without conducting the acyclic check computes all maximum agreement forest for two rooted binary phylogenetic $\cX$-trees, which is a direct consequence from \cite{Albrecht2015b}[Lemma~2]. Moreover, given a relevant maximum agreement forest $\cF$ for $T_1$ and $T_2$, a binary resolution $\hat\cF$ of $\cF$ is automatically a maximum agreement forest corresponding to $T_1(\hat\cF)$ and $T_2(\hat\cF)$. This is, in particular, the case, since just by definition each component $\hat F$ in $\hat\cF$ is a subtree of $T_1(\hat\cF)$ and $T_2(\hat\cF)$ and, as in $\cF$ all components are edge disjoint subtrees in $T_1$ and $T_2$, this has to hold for each of its binary resolution as well. Furthermore, $\hat\cF$ has to be of minimum cardinality, since, otherwise, $\cF$ would not be a maximum agreement forest for $T_1$ and $T_2$. Consequently, by applying the algorithm \scAMSa~to both trees $T_1(\hat\cF)$ and $T_2(\hat\cF)$ the maximum agreement forest $\hat\cF$ is calculated if $k\ge|\cF|$, which, finally, establishes the proof of Lemma~\ref{32-lem-cor0}.
\end{proof}

In the following, we will show that a cherry list $\fcL$ for two binary resolutions of two rooted phylogenetic $\cX$-trees $T_1$ and $T_2$ is also mimicking a computational path of \textsc{allMulMAFs*} applied to $T_1$ and $T_2$.

\begin{lemma}
\label{32-lem-cor1}
Let $\hat T_1$ and $\hat T_2$ be two binary resolutions of two rooted (nonbinary) phylogenetic $\cX$-trees $T_1$ and $T_2$, respectively. Moreover, let $\fcL$ be a cherry list for $\hat T_1$ and $\hat T_2$. Then, $\fcL$ is automatically a cherry list for $T_1$ and $T_2$.
\end{lemma}

\begin{proof}
Lemma~\ref{32-lem-cor1} obviously holds, if the cherry list $\fcL$ only exists of cherry actions $\fc_i=(\{a,c\},\phi_i)$ with $\phi_i\in\{\cup_{ac},\nmid_a,\nmid_c,\}$. This is, in particular, the case because these cherry actions only affect those nodes whose corresponding subtree has been fully contracted so far. When processing a cherry action $\fc_i=(\{a,c\},\phi_i)$ with $\phi_i=\cap_{ac}$, however, two slightly different forests $\hat\cF_{i+1}$ and $\cF_{i+1}$ can arise. More precisely, this is the case, if there is a multifurcating node $x$ lying on the path $P_{ac}$ connecting taxon $a$ and $c$ in $\cF_i$ providing a set $V_x=\{v_0,v_1,v_2,\dots,v_n\}$ of at least $3$ children, where $v_0$ denotes the node which is also part of $P$ and, if, additionally, $R_i$ contains two nodes $d$ and $e$ whose path connecting $d$ and $e$ contains a set of pendant subtrees each corresponding to $R_i(v_i)$, with $i>0$ and $v_i\in V_x$. 

Now, for simplicity, we assume that there is only one such multifurcating node~$x$ of out-degree three so that $V_x=\{v_0,v_1,v_2\}$. In this case, as $\hat\cF_i$ only consists of binary trees, by processing $\fc_i=(\{a,c\},\cap_{ac})$ two components $\hat F_{\alpha}$ and $\hat F_{\beta}$ rooted at $v_1$ and $v_2$, respectively, are added to $\hat\cF_{i+1}$ whereas to $\cF_{i+1}$ only one component $F_{\gamma}$ is added whose root contains two children corresponding to $v_1$ and $v_2$ (cf.~Fig.~\ref{32-fig-psePath}).

\begin{figure}[bt]
\centering
\includegraphics[scale = 1.3]{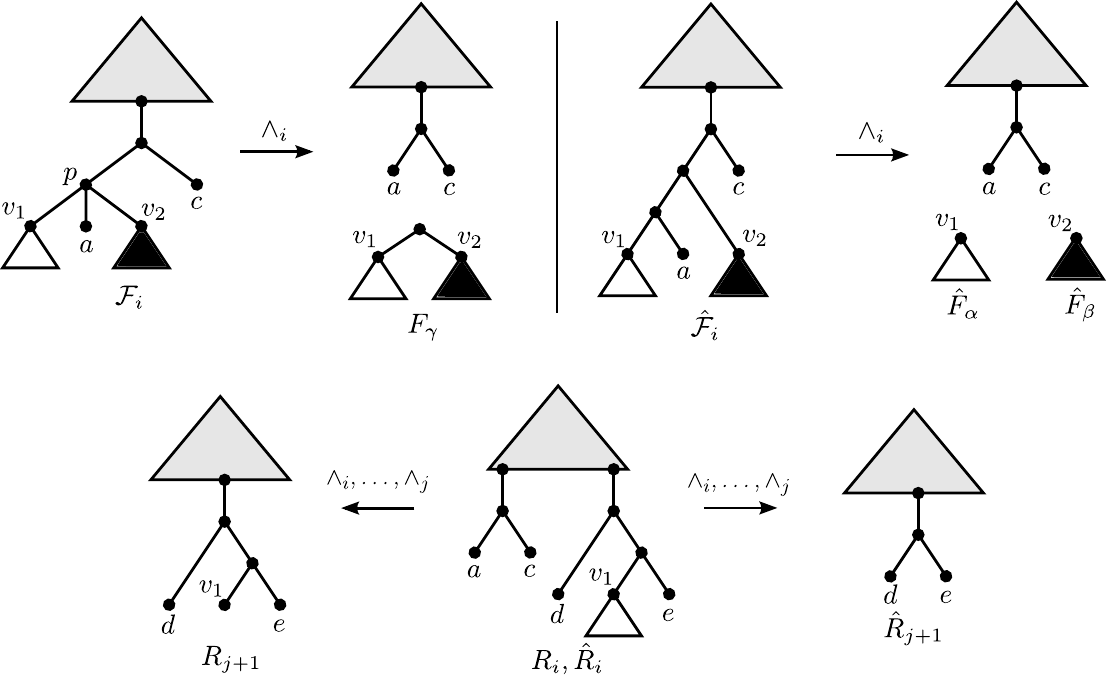}
\caption[An illustration of the scenario regarding Lemma~\ref{32-lem-cor1}]{An illustration of the scenario described in the proof of Lemma~\ref{32-lem-cor1}.} 
\label{32-fig-psePath}
\end{figure} 

Now, let $\fc_j$ be a cherry action in $\fcL$ with $j>i$ in which both components $\hat F_{\alpha}$ and $\hat F_{\beta}$ have been fully contracted so far. As a consequence, since the components $\hat F_{\alpha}$ and $\hat F_{\beta}$ only consist of a single taxon, which has been removed from $\hat R_{j+1}$ (cf.~Alg.~\ref{32-alg-proc}, Line~13), the two taxa $d$ and $e$ are now cherries in $\hat R_{j+1}$ which is, however, not the case in $R_{j+1}$ because $F_{\gamma}$ still contains the two nodes $v_1$ and $v_2$ (cf.~Fig.~\ref{32-fig-psePath}). Nevertheless, since in $F_{\gamma}$ the node $v_1$ and $v_2$ are leaves directly attached to the root, the cherry $\{d,e\}$ is a pseudo cherry in $R_{j+1}$ and, thus, an upcoming cherry action $\fc_k$ containing $\{d,e\}$ represents a pseudo cherry for $R_k$ and $\cF_k$ in this case.

As a consequence, each cherry of $\hat R_i$ and $\hat \cF_i$ corresponding to a cherry action $\fc_i$ in $\fcL$ is either a cherry or a pseudo cherry of $R_i$ and $\cF_i$ and, thus, Lemma~\ref{32-lem-cor1} is established.
\end{proof}

Next, we will show that by expanding a forest $\cF'$ on $\cX$ as prescribed in $M$ derived from calling \textsc{allMulMAFs}, the resulting forest is automatically an agreement forest for both input trees $T_1$ and $T_2$.

\begin{lemma}
Given two rooted (nonbinary) phylogenetic $\cX$-trees $T_1$ and $T_2$, let $\cF$ be a forest on $\cX$ that has been expanded as prescribed in $M$ after $\textsc{allMulMAFs*}(T_1,\{T_2\},\emptyset,k)$ has been called. Then, $\cF$ is an agreement forest for $T_1$ and $T_2$.
\label{32-lem-corX}
\end{lemma}

\begin{proof}
Since, obviously, $\cF$ is a partition of $\cX$, it suffices to consider each of the following cases describing a putative scenario leading to a forest that is not an agreement forest for both input trees $T_1$ and $T_2$, because either the refinement property or the node disjoint property in terms of $T_1$ or $T_2$ is not fulfilled. We will show, however, that during the execution of our algorithm \textsc{allMulMAFs*} each of those scenarios can be excluded.\\

\textbf{Case~1.} Assume there exists a component $F_i$ in $\cF$ such that $F_i$ is not a refinement of $T_2|_{\cX_i}$, where $\cX_i$ denotes the taxa set of $F_i$. As $\cF'$ has been derived from $T_2$ by cutting and contracting its edges, this automatically implies that a cherry has been expanded as prescribed in $M$ in respect of $\top$ instead of $\bot$. However, in $M$ a cherry is only then mapped to $\top$, if and only if, during the $i$-th recursive call, it is a common cherry of $R_i$ and $\cF_i$ and if both parents corresponding to the cherry in $R_i$ and $F_i$ are multifurcating nodes (cf.~Case~2a of \textsc{allMulMAFs*}). Moreover, such a cherry is immediately mapped back to $\bot$, if either the cherry itself or all its siblings have been cut (cf.~Case~1c,2c~of \textsc{allMulMAFs*}). Thus, such a component $F_i$ cannot exist in $\cF$.\\

\textbf{Case~2.} Assume there exists a component $F_i$ in $\cF$ such that $F_i$ is not a refinement of $T_1|_{\cX_i}$, where $\cX_i$ denotes the taxa set of $F_i$. This automatically implies that either a cherry has been expanded as prescribed in $M$ in respect of $\top$ instead of $\bot$ or, during the $i$-th recursive call, a cherry was not a common cherry of $R_i$ and $F_i$. Similar to Case~1, the first scenario can be excluded. Moreover, the latter scenario cannot take place either, since, in order to reduce $F_i$ to a single node, this common cherry must have been contracted which can only take place, if it was a common cherry of $R_i$ and $\cF_i$. Thus, such a component $F_i$ cannot exist in $\cF$.\\ 

\textbf{Case~3.} Assume there exist two components $F_i$ and $F_j$ in $\cF$, with taxa set $\cX_i$ and $\cX_j$, respectively, such that $T_2(\cX_i)$ and $T_2(\cX_j)$ are not edge disjoint in $T_2$. As $F_i$ and $F_j$ must be a refinement of $T_2|_{\cX_i}$ and $T_2|_{\cX_j}$, respectively, (cf.~Case~1) and both components have been derived from $T_2$ by cutting some of its edges, only one of both components can be part of $\cF$. As a direct consequence, such two components cannot exist in $\cF$.\\ 

\textbf{Case~4.} Assume there exist two components $F_i$ and $F_j$ in $\cF$ such that $T_1(\cX_i)$ and $T_1(\cX_j)$ are not edge disjoint in $T_1$. As shown in Case~2, $F_i$ and $F_j$ must be a refinement of $T_1|_{\cX_i}$ and $T_1|_{\cX_j}$, respectively, and, thus, in order to obtain $F_i$ and $F_j$ the following steps must be performed during the execution of \textsc{allMulMAFs*}. Let $E_i$ be an edge set that is only contained in $T_1(\cX_i)$ and not in $T_1(\cX_j)$. In order to obtain $F_i$, some of those edges in $E_j$ must be cut, whereas, in order to obtain $F_j$, all of them must be contracted which leads to a contradiction (cf.~Fig.~\ref{32-fig-edgeDis}). Thus, such two components cannot exist in $\cF$.\\

\begin{figure}[bt]
\centering
\includegraphics[scale = 1.3]{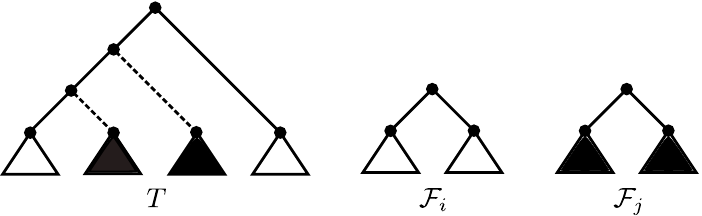}
\caption[An illustration of Case 4 regarding the proof of Lemma~\ref{32-lem-corX}]{An illustration of the scenario described in Case 4 corresponding to the proof of Lemma~\ref{32-lem-corX}. In order to obtain $F_i$ from $T$, both dashed edges have to be cut whereas, in order to obtain $F_j$, these two edges have to be contracted which is a contradiction.} 
\label{32-fig-edgeDis}
\end{figure} 

Finally, by combining all four cases Lemma~\ref{32-lem-corX} is established.
\end{proof}

Moreover, in the following, we will show that each agreement forest $\cF$ that is reported by \textsc{allMulMAFs*} is relevant which means that $\cF$ does not contain an edge $e$ such that $\cF\ominus\{e\}$ is still an agreement forest for both input trees.

\begin{lemma}
Let $T_1$ and $T_2$ be two rooted (nonbinary) phylogenetic trees, then, each agreement forest that is reported by applying \textsc{allMulMAFs*} to $T_1$ and $T_2$ is relevant.
\label{32-lem-cor4}
\end{lemma}

\begin{proof}
Just by definition, given two rooted (nonbinary) phylogenetic trees $T_1$ and $T_2$, an agreement forest $\cF$ for $T_1$ and $T_2$ that is not relevant has to contain an edge $e$ such that $\cF\ominus\{e\}$ is still an agreement forest for $T_1$ and $T_2$. Such an edge $e$, however, can only arise, if a cherry of a multifurcating node is expanded in respect of $\bot$ instead of $\top$. Initially, such a cherry must have been set to $\top$ because during the $i$-th recursive call both corresponding parents in $R_i$ and $\cF_i$, respectively, must have been multifurcating nodes. The only scenario setting $\top$ to $\bot$ would arise, if the cherry itself or all its siblings are cut during subsequent recursive calls. In this case, however, this cherry has to be set to $\bot$, since, otherwise the resulting forest would not be an agreement forest for $T_1$ and $T_2$. Thus, such an edge $e$ cannot exist and, consequently, Lemma~\ref{32-lem-cor4} is established. 
\end{proof}

Since both algorithms \textsc{allMulMAFs*} and \textsc{allMulMAFs} process common cherries and contradicting cherries in the same way, Lemma~\ref{32-lem-cor4}, obviously, has to hold for \textsc{allMulMAFs} as well.

\begin{corollary}
Let $T_1$ and $T_2$ be two rooted (nonbinary) phylogenetic trees, then, each agreement forest that is reported by applying \textsc{allMulMAFs} to $T_1$ and $T_2$ is relevant.
\label{32-cor-cor1}
\end{corollary}

Now, let $T_1$ and $T_2$ be two rooted (nonbinary) phylogenetic $\cX$-trees. By the following Lemma \ref{32-lem-cor2}, we will show that for each maximum binary agreement forest $\hat \cF$, which can be computed by applying \textsc{ProcessCherries} to a cherry list $\fcL$ for two binary resolutions of $T_1$ and $T_2$, by calling $\textsc{ProcessCherries}(T_1,T_2,\fcL)$ a forest $\cF$ is computed such that $\hat\cF$ is a binary resolution of $\cF$.

\begin{lemma}
Given two rooted (nonbinary) phylogenetic $\cX$-trees $T_1$ and $T_2$, let $\hat T_1$ and $\hat T_2$ be two binary resolutions of $T_1$ and $T_2$, respectively. Moreover, let $\hat \cF$ be an agreement forest for $\hat T_1$ and $\hat T_2$ obtained from calling $\textsc{ProcessCherries}(\hat T_1,\{\hat T_2\},\fcL)$, where $\fcL$ denotes a cherry list for $\hat T_1$ and $\hat T_2$. Then, a relevant agreement forest $\cF$ is calculated by calling $\textsc{ProcessCherries}(T_1,\{T_2\},\fcL)$ such that $\hat\cF$ is a binary resolution of $\cF$.
\label{32-lem-cor2}
\end{lemma}

\begin{proof}
We will first show by induction a slightly modified version of Lemma~\ref{32-lem-cor2}. 

Given two rooted (nonbinary) phylogenetic $\cX$-trees $T_1$ and $T_2$. Let $\hat T_1$ and $\hat T_2$ be two binary resolutions of $T_1$ and $T_2$, respectively, and let $\hat \cF_i$ and $\cF_i$ be each forest corresponding to iteration $i$ while executing $$\textsc{ProcessCherries}(\hat T_1,\{\hat T_2\},\fcL)\text{ and }\textsc{ProcessCherries}(T_1,\{T_2\},\fcL),$$ respectively, where $\fcL=(\fc_1,\fc_2,\dots,\fc_n)$ is a cherry list for $\hat T_1$ and $\hat T_2$. Then, $\hat \cF_i$ is a called \emph{pseudo binary resolution} of $\cF_i$, which is defined as follows. Given two forests $\hat\cF$ and $\cF$ for a phylogenetic $\cX$-tree, we say that $\hat\cF$ is a pseudo binary resolution of $\cF$, if for each component $\hat F$ in $\hat\cF$ there exists a component $F$ in $\cF$ such that one of the two following properties hold.

\begin{itemize}
\item[(i)] $\hat F$ is a binary resolution of $F$.
\item[(ii)] $\hat F$ is a binary resolution of $F(v)$, where $v$ is a child of the root of $F$.
\end{itemize} 

The following proof is established by an induction on $i$ denoting the position of a cherry action in $\fcL=(\fc_1,\fc_2,\dots,\fc_n)$.\\

\textbf{Base case.} At the beginning, $F_1$ only consists of $\hat T_2$, which is a binary resolution of $T_2$. Thus, the assumption obviously holds for $i=1$.

\textbf{Inductive step.} Depending on the cherry action $\fc_i=(\{a,c\},E)$, the forest $\hat \cF_{i+1}$ can be obtained from $\hat \cF_i$ in the following ways.
\begin{itemize}
\item[(i)] If $\{a,c\}$ is a pseudo cherry, a set of nodes $V'$ that is attach to the root of a component in $\cF_i$ is cut. Since $\fcL$ is a cherry list for $\hat T_1$ and $\hat T_2$ and, thus, $\{a,c\}$ is a cherry in $\hat R_i$, each of node in $V'$ already refers to components in $\hat\cF_i$ all consisting only of isolated nodes. Thus, after cutting the in-edge of each node in $V'$, $\hat\cF_i$ is still a pseudo binary resolution of $\cF_i$.
\item[(ii)] If $\{a,c\}$ is a common cherry and, thus $\phi_i=\cup_{ac}$, in both forests $\hat \cF_i$ and $\cF_i$ the two taxa $a$ and $c$ are contracted. Consequently, since $\hat \cF_i$ is a pseudo binary resolution of $\cF_i$, this directly implies that $\hat \cF_{i+1}$ is a pseudo binary resolution of $\cF_{i+1}$ as well.
\item[(iii)] If $\{a,c\}$ is a contradicting cherry and $\phi_i=\nmid_a$ (or $\phi_i=\nmid_c$), then, in both forests $\hat \cF_i$ and $\cF_i$ the node labeled by taxon $a$ (or taxon $c$) is cut. Again, no matter if $a\sim c$ or $a\not\sim c$ holds, since $\hat \cF_i$ is a pseudo binary resolution of $\cF_i$, this directly implies that $\hat \cF_{i+1}$ is a pseudo binary resolution of $\cF_{i+1}$ as well.
\item[(iv)] If $\{a,c\}$ is a contradicting cherry and $\phi_i=\cap_{ac}$, in both forests $\hat\cF_i$ and $\cF_i[a\sim c]$ each pendant subtree lying on the path connecting both leaves labeled by $a$ and $c$, respectively, is cut. Let $\hat \cF'$ and $\cF'$ be those component arising from cutting $\hat\cF_i$ and $\cF_i[a\sim c]$, respectively. Since $\hat \cF_i$ is a binary resolution of $\cF_i$, $|\cF'|\ge|\hat\cF'|$ holds which means, in particular, that each $\hat F'_i$ in $\hat\cF'$ is either a binary resolution of $F'_j$ or a binary resolution of $F'_j(v)$ in $\cF'$, where $v$ corresponds to a child whose parent is the root of $F'_j$. Thus, since $\hat \cF_i$ is a pseudo binary resolution of $\cF_i$, this directly implies that $\hat \cF_{i+1}$ is a pseudo binary resolution of $\cF_{i+1}$ as well.
\end{itemize} 
Now, from the induction we can deduce that, independent from the cherry action $\fc_i$, $\hat\cF_i$ is always a pseudo binary resolution of $\cF_i$. Moreover, let $\hat\cF_{n+1}$ and $\cF_{n+1}$ be the two forests obtained from $\hat\cF_n$ and $\cF_n$, respectively, by applying $\fc_n$. Then, since $\hat\cF$ is an agreement forest for $\hat T_1$ and $\hat T_2$, all components in $\hat\cF_{n+1}$ only consist of single isolated nodes which directly implies that $\cF_{n+1}$ does not contain any cherries. Furthermore, due to Lemma~\ref{32-lem-corX}, by expanding $\cF_{n+1}$ as prescribed in $M$ a relevant agreement forest $\cF$ arises such that $\hat\cF$ is a binary resolution of $\cF$ which completes the proof of Lemma~\ref{32-lem-cor2}.
\end{proof}

In the following, we will show that Lemma~\ref{32-lem-cor2} also holds for the original algorithm \textsc{allMulMAFs}.

\begin{lemma}
Given two rooted (nonbinary) phylogenetic $\cX$-trees $T_1$ and $T_2$, let $\hat\cF$ be a binary maximum agreement forest for $T_1$ and $T_2$. Then, by calling $$\textsc{allMulMAFs}(T_1,\{T_2\},\emptyset,k)$$ a relevant maximum agreement forests $\cF$ for $T_1$ and $T_2$ is computed such that $\hat\cF$ is a binary resolution of $\cF$, if and only if $k\ge h(T_1,T_2)$.
\label{32-lem-cor3}
\end{lemma}

\begin{proof}
Notice that, as proven in Lemma~\ref{32-lem-cor2}, Theorem~\ref{32-lem-cor3} holds for the modified algorithm \textsc{allMulMAFs*}. Thus, in order to establish Lemma~\ref{32-lem-cor3}, we just have to show that the following two differences between both algorithms \textsc{allMulMAFs*} and \textsc{allMulMAFs} do not have an impact on the computation of maximum agreement forests.\\

\textbf{Needless cherries.} First of all, let $\fcL$ be a cherry list for $T_1$ and $T_2$ mimicking a computational path of \textsc{allMulMAFs*} resulting in an agreement forest $\cF$ for $T_1$ and $T_2$. Moreover, let $\fcL$ contain a cherry action $\fc_i=(\{a,b\},\phi_i)$ in which $\{a,b\}$ is a contradicting cherry of $R_i$ and $\cF_i$. Now, if $R_i$ contains a taxon $c$ such that $\{a,c\}$ is a common cherry, we call $\{a,b\}$ a \emph{needless cherry}. Notice that a computational path corresponding to the original algorithm \textsc{allMulMAFs} does not consider needless cherries as it always prefers common cherries to contradicting cherries. In the following, however, we will show that for the computation of maximum agreement forests each computational path processing needless cherries can be neglected. 

Let $\cF$ be an agreement forest resulting from a computational path of \textsc{allMulMAFs*} in which, instead of processing a common cherry $\{a,c\}$, a needless cherry $\{a,b\}$ is processed by the cherry action $\fc_i=(\{a,b\},\nmid_a)$. This implies that $\cF$ contains a component $F_a$ corresponding to the expanded taxon $a$, which has been cut during the $i$-th iteration. Moreover, let $F_c$ be the component in $\cF$ containing the node $v_c$ corresponding to taxon $c$ in $\cF_i$. Since $\{a,c\}$ has been a common cherry in iteration $i$, by attaching $F_a$ back to the in-edge of $v_c$ an agreement forest of size $k-1$ arises and, thus, $\cF$ cannot be a maximum agreement forest. Consequently, from cutting instead of contracting common cherries a maximum agreement forest cannot arise and, thus, for the computation of maximum agreement forests each computational path of \textsc{allMulMAFs*} processing needless cherries can be neglected.

Notice that, in this case, the cherry action $\fc_i=(\{a,b\},\nmid_b)$ would be also not be considered. However, after having contracted the common cherry $\{a,c\}$, $b$ could still be cut selecting a cherry action involving one of its siblings.\\

\textbf{Pseudo cherries.} Furthermore, in contrast to the modified algorithm \textsc{allMulMAFs*}, a computational path corresponding to the original algorithm \textsc{allMulMAFs} does not consider pseudo cherries. In the following, however, we will show that for an agreement forest $\cF$ resulting from a computational path processing pseudo cherries, there exists a different computational path calculating $\cF$ without considering any pseudo cherries.

Let $\fcL$ be a cherry list for $T_1$ and $T_2$ mimicking a computational path of \textsc{allMulMAFs*} resulting in an agreement forest $\cF$ and let $\fc_i$ be a cherry action whose corresponding cherry $\{a,c\}$ is a pseudo cherry of $R_i$ and $\cF_i$. Moreover, let be $(b_1,b_2,\dots,b_k)$ and $(b_{k+1},b_2,\dots,b_n)$ be those taxa corresponding to each pendant node lying on the path connecting $a$ and $\scLCA_{R_i}(\{a,c\})$ as well as $c$ and $\scLCA_{R_i}(\{a,c\})$, respectively. Then, we can replace $\fc_i=(\{a,c\},\phi_i)$ through the sequence of cherry actions $$(\{a,b_1\},\nmid_{b_1}),\dots,(\{a,b_k\},\nmid_{b_k}),(\{c,b_{k+1}\},\nmid_{b_{k+1}}),\dots,(\{c,b_n\},\nmid_{b_n}),(\{a,c\},\phi_i)$$ neither containing pseudo cherries nor needles cherries such that the agreement forest $\cF$ is still computed. This means, in particular, that each tree operation that is conducted for preparing a pseudo cherry can be also realized by a sequence of cherry actions neither containing needless cherries nor pseudo cherries.\\

As shown above, for a relevant maximum agreement forest $\cF$ our modified algorithm \textsc{allMulMAFs*} always contains a computational path calculating $\cF$ by neither taking needless cherries nor pseudo cherries into account. Thus, each relevant maximum agreement forest for $T_1$ and $T_2$ that is calculated by \textsc{allMulMAFs*} is also calculated by \textsc{allMulMAFs} and, as a direct consequence, Lemma~\ref{32-lem-cor3} is established.
\end{proof}

Now, in a last step, we can finish the proof of Theorem~\ref{32-th-cor1}. Let $\hat T_1$ and $\hat T_2$ be two binary resolutions of two rooted (nonbinary) phylogenetic $\cX$-trees $T_1$ and $T_2$, respectively, and let $\hat\cF$ be an agreement forest for $\hat T_1$ and $\hat T_2$. Then, by combining Corollary~\ref{32-cor-cor1} and Lemma~\ref{32-lem-cor3} we can deduce that the algorithm \textsc{allMulMAFs} computes a relevant maximum agreement forest $\cF$ for $T_1$ and $T_2$ such that $\hat\cF$ is a binary resolution of $\cF$. This automatically implies, that our algorithm calculates all relevant maximum agreement forests for $T_1$ and $T_2$ and, thus, Theorem~\ref{32-th-cor1} is finally established.
\end{proof}

\subsection{Runtime of \textsc{allMulMAFs}}

In this section, we discuss the theoretical worst-case runtime of the algorithm \textsc{allMulMAFs} in detail.

\begin{theorem}
Let $T_1$ and $T_2$ be two rooted phylogenetic $\cX$-trees and $\cF$ be a relevant maximum agreement forest for $T_1$ and $T_2$ containing $k$ components. The theoretical worst-case runtime of the algorithm \textsc{allMulMAFs} applied to $T_1$ and $T_2$ is $O(3^{|\cX|+k}|\cX|)$.
\label{32-th-rt1}
\end{theorem}

\begin{proof}
Let $\cF=\{F_{\rho},F_1,F_2,\dots,F_{k-1}\}$ be an agreement forest for $T_1$ and $T_2$ of size $k$. To obtain $\cF$ from $T_2$, obviously $k-1$ edge cuttings are necessary. Moreover, in order to reduce the size of the leaf set $\cX$ of $R$ to $1$, to each component $F_i$ in $\cF$ we have to apply exactly $|\cL(F_i)|-1$ cherry contractions. Consequently, at most $|\cX|$ cherry contractions have to be performed in total. Thus, our algorithm has to perform at most $O(|\cX|+k)$ recursive calls for the computation of $\cF$. Now, as one of these recursive calls can at least branch into $3$ further recursive calls, $O(3^{|\cX|+k})$ is an upper bound for the total number of recursive calls that are performed throughout the whole algorithm. Moreover, each case that is conducted during a recursive (cf.~Sec.~\ref{32-sec-algAF}) can be done in $O(|\cX|)$ time and, thus, the theoretical worst-case runtime of the algorithm can be estimated with $O(3^{|\cX|+k}|\cX|)$.
\end{proof}

\subsection{Conclusion}

In this section, we have presented the algorithm \textsc{allMulMAFs} calculating all relevant maximum agreement forests for two rooted (nonbinary) phylogenetic $\cX$-trees. Therefor, we have established a detailed formal proof showing the correctness of the algorithm which is based on both previously presented algorithms \textsc{allMAAFs} and \scAMa. In the next section, we will show how to further modify the algorithm \textsc{allMulMAFs} so that now all relevant maximum \emph{acyclic} agreement forests are calculated. 

\section{The algorithm \textsc{allMulMAAFs}}
\label{34-chap-allMulMAAFs}

In this section, we show how to extend the algorithm \textsc{allMulMAFs}, presented in Section~\ref{32-sec-algAF}, such that the reported agreement forests additional satisfy the acyclic constraint, which automatically implies that the extended algorithm will calculated all relevant maximum \emph{acyclic} agreement forests for two rooted \emph{nonbinary} phylogenetic $\cX$-trees. As mentioned previously, the acyclic constraint plays an important role for the construction of hybridization networks as, for example, demonstrated by the algorithm \textsc{HybridPhylogeny} \cite{Baroni2006}. More specifically, this algorithms generates a hybridization network displaying two rooted \emph{bifurcating} phylogenetic $\cX$-trees from the components of an acyclic agreement forest of those two trees. Thus, we consider the computation of nonbinary maximum acyclic agreement forests as a first step to come up with minimum hybridization networks displaying the refinements of two rooted nonbinary phylogenetic $\cX$-trees.

Broadly speaking, the algorithm \textsc{allMulMAFs} can be used to make progress towards an agreement forest for two rooted (nonbinary) phylogenetic $\cX$-trees $T_1$ and $T_2$ as long as the set of components $\cF$ does not satisfy all properties of an agreement forest. Once our algorithm has successfully computed a maximum agreement forest $\cF$ for $T_1$ and $T_2$, we can apply a specific tool that is able to check, if we can refine $\cF$ to a maximum acyclic agreement forest. Such a refinement of an agreement forest is done by cutting a minimum number of edges within its components such that each directed cycle of the underlying ancestor-descendant graph $AG(T_1,T_2,\cF)$ is dissolved. 

Notice that this problem is closely related to the \emph{directed feedback vertex set problem}. More specifically, given a directed graph $G$ with node set $V$, a feedback vertex set $V'$ is a subset of $V$ containing at least one node of each directed cycle of $G$. This implies, by deleting each node of $V'$ together with its adjacent edges, each directed cycle is automatically removed. Now, based on a directed graph, the directed feedback vertex set problem consists of minimizing the size of such a feedback vertex set.

\subsection{Refining agreement forests}

In this section we present a tool that enables the refinement of an agreement forest. We call this tool an \emph{expanded ancestor-descendant graph}. Notice that this tool has been previously published under a different term as we state in the following.

\begin{remark}
The following concept of an \emph{expanded ancestor-descendant graph} corresponds to the concept of an \emph{expanded cycle graph} given in the work of Whidden \textit{et al.} \cite{Whidden2011}. The latter concept, however, can be only applied to agreement forests corresponding to rooted \emph{binary} phylogenetic $\cX$-trees. Hence, we have adapted this concept such that it can be also applied to agreement forests corresponding to rooted \emph{nonbinary} phylogenetic $\cX$-trees. Notice that, adapting the concept of an expanded cycle graph to nonbinary agreement forests has been also examined in the master thesis of Li \cite{Li2014}. In this work, each step that is necessary to compute the hybridization number for two rooted (nonbinary) phylogenetic $\cX$-trees is presented in detail by, additionally, discussing its correctness.
\end{remark}

In the following, we give a short overview of how an expanded ancestor-descendant graph is defined and how this graph can be used to transform agreement forests into acyclic agreement forests.\\

\textbf{Expanded ancestor-descendant graph.} The tool that enables the refinement of an agreement forest $\cF$ for two rooted (nonbinary) phylogenetic $\cX$-trees $T_1$ and $T_2$ to an acyclic agreement forest is an \emph{expanded ancestor-descendant graph} $AG^\text{ex}(T_1(\cF),T_2(\cF),\cF)$. In contrast to the ancestor-descendant graph, each node of this graph corresponds to exactly one particular node of a component in $\cF$. Thus, from such a graph one can directly figure out those edges of a component that have to be cut in order to remove a directed cycle (cf.~Fig.~\ref{33-fig-ExpAG}). 

Given two rooted phylogenetic $\cX$-trees $T_1$ and $T_2$ and a nonbinary maximum acyclic agreement forest $\cF$ for $T_1$ and $T_2$, the corresponding expanded ancestor-descendant graph $AG^\text{ex}(T_1(\cF),T_2(\cF),\cF)$ consists of the following nodes and edges. First of all, $\cF$ is a subset of $AG^\text{ex}(T_1(\cF),T_2(\cF),\cF)$, which means that the graph contains all nodes and edges corresponding to all components in $\cF$. Moreover, $AG^\text{ex}(T_1(\cF),T_2(\cF),\cF)$ contains a set of \emph{hybrid edges} each connecting two specific nodes each being part of two different components. More precisely, those edges are defined as follows. 

Given a node $v$ of a component $F_j$ in $\cF$, the function $\phi_i(v)$ refers to the lowest common ancestor in $T_i(\cF)$, with $i\in\{1,2\}$, of each leaf that is labeled by a taxon contained in $\cL(F_j(v))$. Notice that the node $\phi_i(\cdot)$ is well defined, which means there exists exactly one node in $T_i(\cF)$ to which $\phi_i(\cdot)$ applies. Equivalently, the function $\phi_i^{-1}(\cdot)$ maps nodes from $T_i$, with $i\in\{1,2\}$, back to a component in $\cF$. More precisely, let $E_i$, with $i\in\{1,2\}$, be the set of edges consisting of all in-edges of all lowest common ancestors in $T_i$ of the taxa set of each $F_j$ in $\cF\setminus\{F_{\rho}\}$. Then, the node $v\in T_i$ maps back to the node in $\cF$ representing the lowest common ancestor of those taxa that can be reached from $v$ by not using an edge in $E_i$. Notice, however, that this function is only defined for those nodes that are either labeled or are part of a path connecting two labeled nodes $a$ and $b$ such that $\phi_i^{-1}(a)$ and $\phi_i^{-1}(b)$ are contained in the same component $F_j$ in $\cF$. Similar to the binary case, since the graph is built for the trees $T_1(\cF)$ and $T_2(\cF)$ reflecting $\cF$, the function $\phi_i^{-1}(\cdot)$ is well defined, which means that, if defined, there exists exactly one node in $\cF$ to which $\phi_i^{-1}(\cdot)$ applies. 

Now, based on the definitions of these two functions, $AG^\text{ex}(T_1(\cF),T_2(\cF),\cF)$ contains the following hybrid edges. Let $w$ be a node in this graph corresponding to the root of a component $F_j$ not equal to $F_{\rho}$. Moreover, for the tree $T_i(\cF)$ with $i\in\{1,2\}$, let $v'$ be the lowest ancestor of $\phi_i(w)$ such that $\phi_i^{-1}(v')$ is defined. In more detail, let $P_{\phi}=(v_1,\dots,v_n)$ be those nodes lying on the path connecting the parent $v_1$ of $v$ and the root $v_n$ of $T_1$ such that $v_j$ with $j\in[2:n]$ is the parent of $v_{j-1}$. Then, $$v'=\min_j\{v_j:v_j \in P_{\phi} \wedge \phi_i^{-1}(v_j) \text{ is defined}\}.$$ Based on $v'$ and $w$, $AG^\text{ex}(T_1(\cF),T_2(\cF),\cF)$ contains a hybrid edge $(\phi_i^{-1}(v'),w)$. Notice that, if $\cF$ contains $k$ components, for each component except $F_{\rho}$ two hybrid edges corresponding to $T_1$ and $T_2$ are inserted which are $2k-2$ hybrid edges in total. Furthermore, the target node of a hybrid edge does always refer to a root node of a component $F_j$ in $\cF$ whereas the source node never does.\\

\textbf{Exit nodes.} Given two rooted phylogenetic $\cX$-trees $T_1$ and $T_2$ as well as a nonbinary maximum acyclic forest $\cF$ for $T_1$ and $T_2$, an \emph{exit node of $AG^\text{ex}(T_1(\cF),T_2(\cF),\cF)$} is defined as follows. Let $H_i$ be the set of hybrid edges in $AG^\text{ex}(T_1(\cF),T_2(\cF),\cF)$ resulting from $T_i$ with $i\in\{1,2\}$. Now, given a directed cycle in $AG^\text{ex}(T_1(\cF),T_2(\cF),\cF)$ running through the hybrid edges $E_H=\{h_0,\dots,h_{n-1}\}$ in sequential order, then, the source node $v_i$ of a hybrid edge $h_i=(v_i,w_i)$ in $E_H$ is called an exit node, if $h_i$ is contained in $H_1$ and $h_{j}$, with $j=(i-1)\mod{n}$, is contained in $H_2$ or vice versa.\\

Now, based on an expanded ancestor-descendant graph we can refine an agreement forest by \emph{fixing its exit nodes}. An exit node $v$ belonging to the component $F_j$ is fixed by cutting each edge lying on the path connecting $v$ with the node referring to the root node of $F_j$. Notice that by cutting $k$ of those edges, the resulting agreement forest $\cF'$ consists of $|\cF|+k$ components. 

\begin{figure}[bt]
\centering
\includegraphics[scale = 1.15]{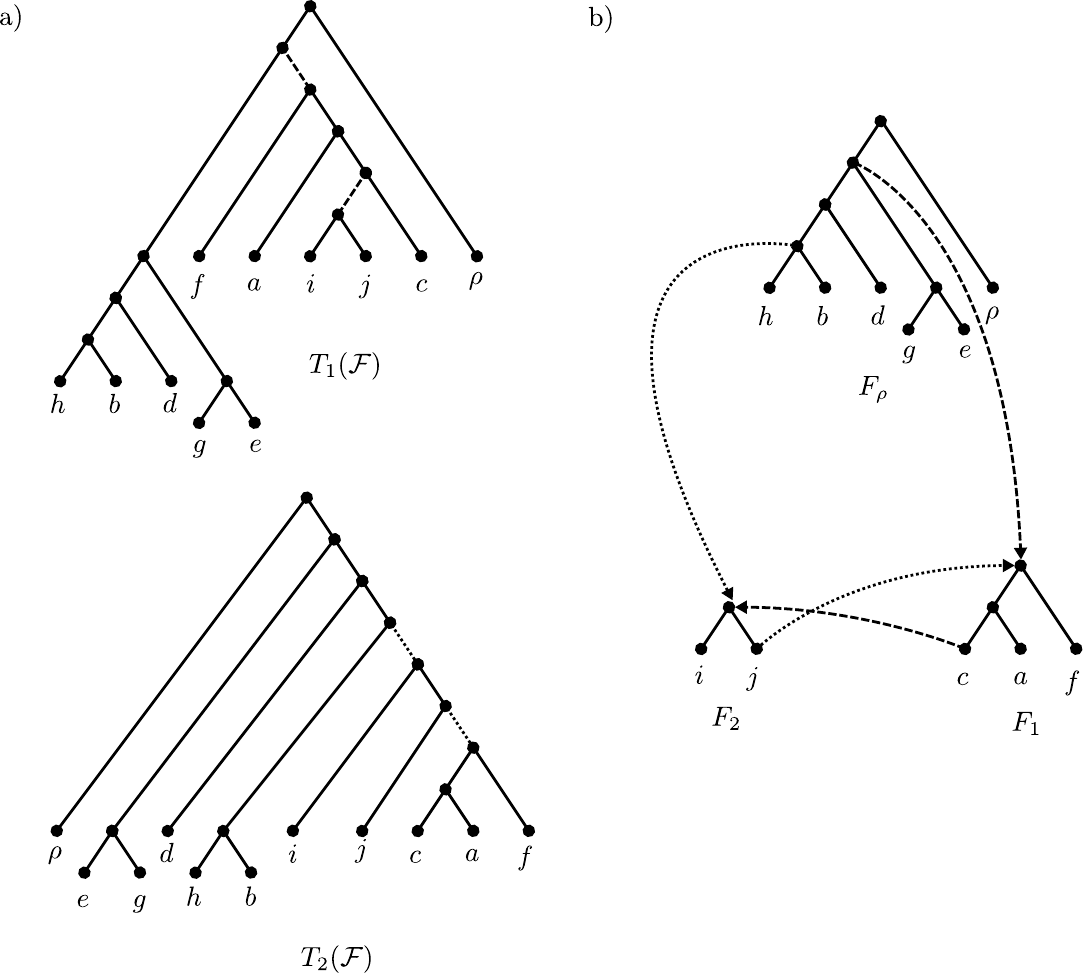}
\caption[An example of an expanded ancestor-descendant graph]{\textbf{(a)} The same two trees $T_1(\cF)$ and $T_2(\cF)$ as depicted in Figure~\ref{31-fig-TRef}. Here, the set of dashed edges and the set of dotted edges refers to the in-edges of the nodes corresponding to the lowest common ancestors of $\cL(F_1)$ and $\cL(F_2)$. \textbf{(b)}The expanded ancestor-descendant graph $AG^\text{ex}(T_1(\cF),T_2(\cF),\cF)$ with $\cF=\{F_{\rho},F_1,F_2\}$. Dashed edges are hybrid edges resulting from $T_1(\cF)$ and dotted edges are hybrid edges resulting from $T_2(\cF)$. Here, for a better overview, the directions of edges corresponding to components of $\cF$ are omitted. Notice that, by fixing the exit node corresponding to taxon $j$, all directed cycles are removed and a maximum acyclic agreement forest for $T_1$ and $T_2$ with size $4$ arises.} 
\label{33-fig-ExpAG}
\end{figure}

\subsection{The algorithm}

We can easily turn the algorithm \textsc{allMulMAFs} into the algorithm \textsc{allMulMAAFs} by applying a post-processing step refining agreement forests. More precisely, given an agreement forest $\cF$ for two rooted phylogenetic $\cX$-trees $T_1$ and $T_2$, by applying the following refinement procedure only those relevant acyclic agreement forests are returned whose size is smaller than or equal to $k$.

\begin{itemize}
\item[(1)] Compute two trees $T_1(\cF)$ and $T_2(\cF)$ reflecting $\cF$.
\item[(2)] Build the expanded ancestor-descendant graph $AG^\text{ex}(T_1(\cF),T_2(\cF),\cF)$.
\item[(3)] Compute the set of exit nodes $V_H$ of $AG^\text{ex}(T_1(\cF),T_2(\cF),\cF)$.
\item[(4)] For each exit node $v$ in $V_H$ turn $\cF$ into $\cF'$ by fixing $v$.
\item[(5)] For each agreement forest $\cF'$ with $|\cF'|\le k$ continue with step 5a or 5b.
\begin{itemize}
\item[(5a)] If $\cF_i'$ is acyclic, return $\cF'$.
\item[(5b)] Otherwise, if $\cF'$ is \emph{not} acyclic, repeat step 2--5 with $\cF'$.
\end{itemize}
\end{itemize}

Based on these steps, by modifying Case~1b as follows, we can easily turn the algorithm \textsc{allMulMAFs} into an algorithm computing a set of maximum acyclic agreement forests.\\

\textbf{Case 1b'.} If $R$ only consists of a single leaf, first each $F_i$ in $\cF$ is expanded as prescribed in $M$ and then $\cF$ is refined with the help of $AG^\text{ex}(T_1(\cF),T_2(\cF),\cF)$ into $\cF'$. Finally, $\cF'$ is returned.\\

This means that, each time before reporting an agreement forest $\cF$, we first check, if we can refine $\cF$ to an acyclic agreement forest $\cF'$ of size smaller than or equal to $k$. If this is possible, we return $\cF'$, else, we return the empty set. 

\subsection{Correctness of \textsc{allMulMAAFs}}
\label{33-sec-cor2}

In this section, we show that by applying the presented algorithm \textsc{allMulMAAFs} one can calculate all relevant maximum acyclic agreement forests for two rooted (nonbinary) phylogenetic $\cX$-trees.

\begin{theorem}
Given two rooted (nonbinary) phylogenetic $\cX$-trees, by calling $$\textsc{allMulMAAFs}(T_1,\{T_2\},\emptyset,k)$$ all relevant maximum acyclic agreement forests for $T_1$ and $T_2$ are calculated, if and only if $k\ge h(T_1,T_2)$.
\label{33-lem-maafs}
\end{theorem}

\begin{proof}
The correctness of the algorithm as stated in Theorem~\ref{33-lem-maafs} directly depends on the following two Lemmas~\ref{33-lem-e1}~and~\ref{33-lem-e2}.

\begin{lemma}
Let $T_1$ and $T_2$ be two rooted (nonbinary) phylogenetic $\cX$-trees and let $\cF$ be a relevant maximum acyclic agreement forest $T_1$ and $T_2$. Then, a relevant agreement forest $\cF'$ by calling $\textsc{allMulMAFs}(T_1,T_2,\emptyset,h(T_1,T_2))$ is calculated that can be turned into $\cF$ by first resolving some of its multifurcating nodes and then by cutting some of its edges.
\label{33-lem-e1}
\end{lemma}

\begin{proof}
As the first point holds for the algorithm \scAMc~\cite{Albrecht2015b}[Theorem~3], from Lemma~\ref{32-lem-cor2} we can deduce that this has to hold for the algorithm \textsc{allMulMAFs} as well. More precisely, let $T_1$ and $T_2$ be two rooted (nonbinary) phylogenetic $\cX$-trees and let $\hat\cF$ be a binary agreement forest for $T_1(\cF)$ and $T_2(\cF)$ that can be turned into a maximum acyclic agreement forest for $T_1(\cF)$ and $T_2(\cF)$ by cutting some of its edges. Then, due to Lemma~\ref{32-lem-cor2}, by calling \scAMc$(T_1,T_2,\emptyset,h(T_1,T_2))$ a relevant acyclic agreement forest $\cF$ for $T_1$ and $T_2$ is calculated such that $\hat\cF$ is a binary resolution of $\cF$. Moreover, as $\hat\cF$ can be turned into a maximum acyclic agreement forest by cutting some of its edges $\hat E$, $\cF$ can be turned into a relevant maximum acyclic agreement forest as well by first resolving some of its nodes and then by cutting a certain edge set $E$ with $|\hat E|=|E|$. More specifically, for each edge $\hat e$ in $\hat E$ there exists an edge $e$ that can be obtained from $\cF$ as follows. Let $\hat e=(\hat v,\hat w)$ be an edge in $\hat E$ of a component $\hat F$ in $\hat \cF$, then, as $\hat \cF$ is a binary resolution of $\cF$, $\cF$ has to contain a component $F$ with node $w'$ such that $\cL(\hat F(\hat w))\subseteq\cL(F(w'))$. Now, $e$ is the in-edge of a node $w$ that can be obtained from resolving node $w'$ such that $\cL(F(\hat w))=\cL(F(w))$.
\end{proof}

\begin{lemma}
Given two rooted (nonbinary) phylogenetic $\cX$-trees $T_1$ and $T_2$ as well as a relevant agreement forest $\cF$ for $T_1$ and $T_2$, the refinement step resolves a minimum number of nodes and cuts a minimum number of edges such that $\cF$ is turned into all relevant acyclic agreement forests of minimum size.
\label{33-lem-e2}
\end{lemma}

\begin{proof}
Due to the following two observations that are both discussed in the master thesis of Li \cite{Li2014}, the refinement procedure, which is based on fixing exit nodes as described above, leads to the computation of acyclic agreement forests.

\begin{observation}
Let $\cF$ be an agreement forest for two rooted phylogenetic $\cX$-trees $T_1$ and $T_2$ and let $\cF'$ be an agreement forest that is produced by fixing an exit node of $AG^\text{ex}(T_1,T_2,\cF)$. Then, the set of exit nodes corresponding to $AG^\text{ex}(T_1,T_2,\cF')$ is a subset of the set of exit nodes corresponding to $AG^\text{ex}(T_1,T_2,\cF)$.
\label{33-obs-exp1}
\end{observation}

\begin{observation}
Given an agreement forest $\cF$ for two rooted phylogenetic $\cX$-trees, there exists an acyclic agreement forest $\cF'$, if and only if there exists a set of exit nodes such that fixing theses nodes leads to the computation of $\cF'$.
\label{33-obs-exp2}
\end{observation}

A formal proof showing the correctness of these two observations can be looked up in the master thesis of Li~\cite{Li2014}. More precisely, Observation~\ref{33-obs-exp1} is a consequence of \cite[Lemma~12~and~13]{Li2014}, which ensures that by fixing an exit node one makes progress towards an acyclic agreement forest, and Observation~\ref{33-obs-exp2} is a consequence of  \cite[Lemma~10]{Li2014}, which ensures that it is possible to obtain all relevant maximum acyclic agreement forests from applying the refinement procedure. Notice that, as by fixing exit nodes a minimum number of nodes are resolved and a minimum number of edges are cut, each resulting maximum acyclic agreement forest is automatically relevant. 
\end{proof}

Now, from those two separate proofs each regarding two successive parts, namely the computation of specific relevant agreement forests followed by the refinement procedure establishing the acyclicity of each those forests, we can finally finish the proof of Theorem~\ref{33-lem-maafs}. 
\end{proof}

\subsection{Runtime of \textsc{allMulMAAFs}}

In this section, we discuss the runtime of the algorithm \textsc{allMulMAAFs} in detail.

\begin{theorem}
Let $T_1$ and $T_2$ be two rooted phylogenetic $\cX$-trees and $\cF$ be a maximum agreement forest for $T_1$ and $T_2$ containing $k$ components. The theoretical worst-case runtime of the algorithm \textsc{allMulMAAFs} applied to $T_1$ and $T_2$ is $O(3^{|\cX|+k}4^k|\cX|)$.
\end{theorem}

\begin{proof}
As stated in Theorem~\ref{32-th-rt1}, the algorithm has to conduct $O(3^{|\cX|+k})$ recursive calls. Potentially, for each of those recursive calls we have to apply a refinement step whose theoretical worst-case runtime can be estimated as follows. First notice that the order of fixing exit nodes is irrelevant. Thus, in an expanded ancestor-descendant graph, corresponding to an agreement forest of size $k+1$ and, hence, containing $2k$ exit nodes, at most $2^{2k}$ different sets of potential exit nodes have to be considered. As the processing of such a set of potential exit nodes takes $O(|\cX|)$ time, the theoretical worst-case runtime of the algorithm is $O(3^{|\cX|+k}4^k|\cX|)$. 
\end{proof}

In general, however, due to the following observation, the runtime of the refinement step is not a problem when computing maximum acyclic agreement forests of size $k$. Either the size $k'$ of an agreement forest $\cF$ is close to $k$ and, thus, fixing an exit node immediately leads to an agreement forest of size larger than $k$ (and, consequently, most of the sets of potential exit nodes have not to be considered in full extend). Otherwise, if the size $k'$ of $\cF$ is small and, thus, the gap between $k'$ and $k$ is large, the expanded ancestor-descendant graph is expected to contain no or at least only less cycles (and, consequently, there exist only few sets of potential exit nodes). Nevertheless, in the master thesis of Li \cite{Li2014} a method is presented that allows to half the number of exit nodes that have to be taken into account throughout the refinement of an agreement forest, so that by applying this modification the algorithm yields a theoretical worst-case runtime of $O(3^{|\cX|+k}2^kk)$.

\subsection{Robustness of our Implementation}
\label{sec-rob}

In order to make the algorithm available for research, we added an implementation to our Java based software package \textsc{Hybroscale} providing a graphical user interface, which enables a user friendly interactive handling. Next, we conducted two specific test scenarios demonstrating the robustness of our implementation which means, in particular, that \textsc{Hybroscale} guarantees the computation of all relevant nonbinary acyclic agreement forests for two rooted (nonbinary) phylogenetic $\cX$-trees. Each of those test scenarios was conducted on a particular synthetic dataset, which was generated as described below.

\subsubsection{Synthetic dataset}
\label{33-sec-syndat}

Our synthetic dataset consists of several tree sets each containing two rooted (nonbinary) phylogenetic $\cX$-trees. Each $\cX$-tree is generated by ranging over all different combinations of four parameters, namely the number of leaves $\ell$, an upper bound for the hybridization number $k$, the \emph{cluster degree} $c$, and an additional parameter $p$. Each of both trees of a particular tree set corresponds to an embedded tree $T$ of a particular network $N$ only containing hybridization nodes of in-degree $2$. With respect to the four different parameters such a tree $T$ is computed as follows. First a random binary tree $\hat T$ containing $\ell$ leaves is computed. This is done, in particular, by randomly selecting two nodes $u$ and $v$ of a specific set $V$, which is initialized by creating $\ell$ nodes of both in- and out-degree $0$. The two selected nodes $u$ and $v$ are then connected to a new node $w$. Finally, $V$ is updated by replacing $u$ and $v$ by its parent node $w$. This is done until $V$ only consists of one node corresponding to the root of $\hat T$. In a second step, $k$ reticulation edges are inserted in $\hat T$ with respect to parameter $c$ such that the resulting network $N$ contains precisely $k$ reticulation nodes of in-degree $2$. Finally, after extracting a binary $T'$ from $N$, based on parameter $p$, a certain percentage of its edges are contracted such that a nonbinary tree $T$ is obtained from $T'$. 

In this context, the \emph{cluster degree} is an \emph{ad hoc} concept influencing the computational complexity of a tree set similar to the concept of the \emph{tangling degree} first presented in the work of Albrecht \textit{et al.} \cite{Albrecht2011} (cf.~Fig.~\ref{33-fig-clusDeg}). When adding a reticulation edge $e$ with target node $v_2$ and source node $v_1$, we say that $e$ respects the cluster degree $c$, if $v_1$ cannot be reached from $v_2$ and there is a path of length less than or equal to $c$ leading from $v_2$ to a certain node $p$ such that $v_1$ can be reached from $p$. This means, in particular, that networks respecting a small cluster degree, in general, contain more minimum common clusters than networks respecting a large cluster degree and, thus, often provide a smaller computational complexity when applying a cluster reduction beforehand. 

\begin{figure}[bt]
\centering 
\includegraphics[scale=1]{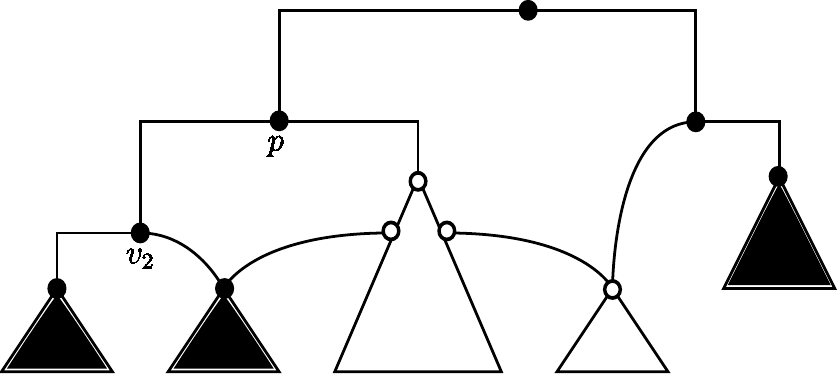}
\caption[An illustration of the cluster degree parameter]{An illustration of the cluster degree parameter. Given a cluster degree $c=1$. When inserting an in-going edge $e$ to node $v_2$ that is respecting $c$, each node that is marked white or is part of a white marked subnetwork forms a potential source node.} 
\label{33-fig-clusDeg}
\end{figure}

\subsubsection{Comparison with other software}

First, we generated a synthetic dataset, as described above, containing tree pairs each consisting of two rooted (nonbinary) phylogenetic $\cX$-trees with parameters $\ell \in \{10,25,50\}$, $k \in \{5,10,15\}$, $c \in \{1,3,5\}$, and $p \in \{30\}$. More specifically, for all $81$ combinations of the four parameters $30$ tree sets were generated resulting in $810$ tree sets in total. Next, based on this dataset, we compared the result of our implementation to the two software packages \textsc{Dendroscope}\footnote{\url{ab.inf.uni-tuebingen.de/software/dendroscope/}} \cite{Huson2011} and \textsc{TerminusEst}\footnote{\url{skelk.sdf-eu.org/terminusest/}} \cite{Piovesan2014} so far being the only known available software packages computing exact hybridization numbers for two rooted (nonbinary) phylogenetic $\cX$-trees.

Our simulation study pointed out that our implementation could always reproduce the hybridization numbers that were computed by both software packages \textsc{Dendroscope} and \textsc{TerminusEst}. Moreover, the number of maximum acyclic agreement forests computed by our algorithm was always larger than the number of networks that were reported by \textsc{Dendroscope} and \textsc{TerminusEst}. Notice that, regarding \textsc{TerminusEst}, this is not surprising as this program does only output one network. This fact, however, gives further indication that our program is actually able to compute all relevant maximum acyclic agreement forests. Nevertheless, we applied a further test scenario examining this fact in more detail. 

\subsubsection{Permutation test}

To check the robustness of our implementation in more detail, we generated a further synthetic dataset containing thousands of tree pairs of low computational complexity, such that each of those tree pairs could be processed by our implementation within less than a minute. More precisely, the dataset contains tree pairs that have been generated in respect to precisely one value for each of the four parameters $\ell$, $k$, $c$, and $p$, i.e., $\ell \in \{10\}$, $k \in \{5\}$, $c \in \{1\}$, and $p \in \{30\}$. Next, for each of those tree pairs, we computed two sets of relevant maximum acyclic agreement forests each corresponding to one of both orderings of the two input trees and compared both results. 

For each of those tree pairs, both sets of maximum acyclic agreement forests were identical, which means that each maximum acyclic agreement forest that could be computed was always contained in both sets. Notice that by switching the order of the input trees our algorithm runs through different recursive calls, which means that each computational path leading to a maximum acyclic agreement forest usually differs. Nevertheless, due to the fact that the hybridization number is independent from the order of the input trees, those two sets of maximum acyclic agreement forests have to be identical. As we applied this permutation test to thousands of different tree pairs, this is a further strong indication that \textsc{Hybroscale} is actually able to compute all relevant maximum acyclic agreement forests for two rooted (nonbinary) phylogenetic $\cX$-trees.

\subsection{Conclusion}

In this section, we have presented the algorithm \textsc{allMulMAAFs} computing a set of relevant maximum acyclic agreement forests for two rooted nonbinary phylogenetic $\cX$-trees. \textsc{allMulMAAFs} was developed in respect to the algorithm \textsc{allHNetworks} computing a particular set of minimum hybridization networks for two rooted binary phylogenetic $\cX$-trees and is considered to be a first step for making this algorithm accessible to nonbinary phylogenetic $\cX$-trees. Additionally, we have established a formal proof showing that the algorithm \textsc{allMulMAAFs} guarantees the computation of all relevant nonbinary maximum acyclic agreement forests. 

Moreover, we have integrated our algorithm into the freely available software package \textsc{Hybroscale} and, by conducting two specific test scenarios, we have demonstrated the robustness of our implementation. In the next section, we will demonstrate how this algorithm can be used to extend the algorithm \textsc{allHNetworks} so that now minimum hybridization networks displaying the refinements of multiple rooted nonbinary phylogenetic $\cX$-trees can be calculated.

\section{Discussion}

In this work, we have presented the algorithm \textsc{allMulMAAFs} computing a set of relevant maximum acyclic agreement forests for two rooted nonbinary phylogenetic $\cX$-trees. \textsc{allMulMAAFs} was developed in respect to the algorithm \textsc{allHNetworks} computing a certain set of minimum hybridization networks for two rooted binary phylogenetic $\cX$-trees and is considered to be the first step to make this algorithm accessible to nonbinary phylogenetic $\cX$-trees. Additionally, we have provided formal proofs showing that the algorithm \textsc{allMulMAAFs} always guarantees the computation of all relevant nonbinary maximum acyclic agreement forests. Moreover, we have integrated our algorithm into the freely available software package \textsc{Hybroscale} and by conducting particular test scenarios, we have demonstrated the robustness of our implementation. It is part of ongoing future work to push on the extension of the algorithm \textsc{allHNetworks} in order to enable the computation of minimum hybridization networks displaying the refinements of multiple rooted nonbinary phylogenetic $\cX$-trees.


\begin{thebibliography}{1}

\bibitem{Albrecht2011}
\textsc{B. Albrecht, C. Scornavacca, A. Cenci, D. H. Huson,} 
\textit{Fast computation of minimum hybridization networks,}
Bioinformatics, 28 (2011), pp. 191--197. 

\bibitem{Albrecht2014}
\textsc{B. Albrecht,}
\textit{Computing Hybridization Networks for Multiple Rooted Binary Phylogenetic Trees by Maximum Acyclic Agreement Forests,}
preprint, arXiv:1408.3044, 2014.

\bibitem{Albrecht2015}
\textsc{B. Albrecht,}
\textit{Computing all hybridization networks for multiple binary phylogenetic input trees,}
BMC Bioinformatics, \textbf{16}:236.

\bibitem{Albrecht2015b}
\textsc{B. Albrecht,}
\textit{Fast computation of all maximum acyclic agreement forests,}
submitted to arXiv.

\bibitem{Baroni2005}
\textsc{M. Baroni, S. Gruenewald, V. Moulton, C. Semple,}
\textit{Bounding the number of hybridisation events for a consisten evolutionary history,}
Mathematical Biology, 51 (2005), pp. 171--182.

\bibitem{Baroni2006}
\textsc{M. Baroni, C. Semple,}
\textit{Hybrids in real time,}
Syst. Biol., 55 (2006), pp. 46--56.

\bibitem{Bordewich2007}
\textsc{M. Bordewich, C. Semple,}
\textit{Computing the Hybridization Number of Two Phylogenetic Trees Is Fixed-Parameter Tractable,}
IEEE/ACM Trans. Comput. Biol. Bioinformatics, 4 (2007), pp. 458--466.

\bibitem{Bordewich2007b}
\textsc{M. Bordewich, C. Semple,}
\textit{Computing the minimum number of hybridization events for a consistent evolutionary history,}
Discrete Applied Mathematics, 155 (2007), pp. 914--928.

\bibitem{Huson2007}
\textsc{D. H. Huson, R. Rupp, C. Scornavacca,}
\textit{Phylogenetic Networks: Concepts, Algorithms and Applications,}
Cambridge University Press, 2011.

\bibitem{Huson2011}
\textsc{D. H. Huson, D. Richter, C. Rausch, T. Dezulian, M. Franz, R. Rupp R.,} 
\textit{Dendroscope: An interactive viewer for large phylogenetic trees,}
BMC Bioinformatics, 8 (2007), pp. 460.

\bibitem{Iersel2014}
\textsc{L. J. J. van Iersel, S. Kelk, N. Lekic, N, S. Leen,}
\textit{Approximation algorithms for nonbinary agreement forests.}
SIAM J. Discrete Math., 28 (2014), pp. 49--66.

\bibitem{Li2014} 
\textsc{Z. Li,}
\textit{Fixed-Parameter Algorithm for Hybridization Number of Two Multifurcating Trees,}
Master thesis, Dalhousie University, Halifax, Canada, 2014.

\bibitem{Piovesan2014} 
\textsc{T. Piovesan, S. Kelk,}
\textit{A simple fixed parameter tractable algorithm for computing the hybridization number of two (not necessarily binary) trees,}
preprint, arXiv:1207.6090, 2014.

\bibitem{Scornavacca2012}
\textsc{C. Scornavacca, S. Linz, L., B. Albrecht,}
\textit{A first step towards computing all hybridization networks for two rooted binary phylogenetic trees,}
J. Comp. Biol., 19 (2012), pp. 1227--1243. 

\bibitem{Whidden2011} 
\textsc{C. Whidden, R. Beiko, N. Zeh,} 
\textit{Fixed-Parameter and Approximation Algorithms for Maximum Agreement Forests,}
SIAM J. Comput., 42 (2011), 1431--1466.

\bibitem{Whidden2014} 
\textsc{C. Whidden, R Beiko, N. Zeh,}
\textit{Fixed-Parameter and Approximation Algorithms for Maximum Agreement Forests of Multifurcating Trees,} 
preprint, arXiv:1305.0512, 2014.

\end{thebibliography}
\end{document}